\documentclass[12pt,english]{article}
\usepackage[T1]{fontenc}
\usepackage[latin9]{inputenc}
\usepackage[letterpaper]{geometry}
\usepackage{units}
\usepackage{textcomp}
\usepackage{url}
\usepackage{amsthm}
\usepackage{amsmath}
\usepackage{amssymb}

\makeatletter

\newcommand{\lyxmathsym}[1]{\ifmmode\begingroup\def\b@ld{bold}
  \text{\ifx\math@version\b@ld\bfseries\fi#1}\endgroup\else#1\fi}

\providecommand{\tabularnewline}{\\}

\theoremstyle{plain}
\newtheorem{thm}{Theorem}
  \theoremstyle{remark}
  \newtheorem{rem}[thm]{Remark}
  \theoremstyle{plain}
  \newtheorem{prop}[thm]{Proposition}

\usepackage{url}

\usepackage{babel}

\makeatother

\usepackage{babel}

\begin{document}

\title{Series crimes}

\author{David R. Stoutemyer%
\thanks{dstout at hawaii dot edu%
}}

\maketitle

\begin{abstract}
Puiseux series are power series in which the exponents can be fractional
and/or negative rational numbers. Several computer algebra systems
have one or more built-in or loadable functions for computing truncated
Puiseux series. Some are generalized to allow coefficients containing
functions of the series variable that are dominated by any power of
that variable, such as logarithms and nested logarithms of the series
variable. Some computer algebra systems also have built-in or loadable
functions that compute \textsl{infinite} Puiseux series. Unfortunately,
there are some little-known pitfalls in computing Puiseux series.
The most serious of these is expansions within branch cuts or at branch
points that are incorrect for some directions in the complex plane.
For example with each series implementation accessible to you:

Compare the value of $(z^{2}+z^{3})^{3/2}$ with that of its truncated
series expansion about $z=0$, approximated at $z=-0.01$. Does the
series converge to a value that is the negative of the correct value?

Compare the value of $\ln(z^{2}+z^{3})$ with its truncated series
expansion about $z=0$, approximated at $z=-0.01+0.1i$. Does the
series converge to a value that is incorrect by $2\pi i$?

Compare $\mathrm{arctanh}(-2+\ln(z)z)$ with its truncated series
expansion about $z=0$, approximated at $z=-0.01$. Does the series
converge to a value that is incorrect by about $\pi i$?

At the time of this writing, most implementations that accommodate
such series exhibit such errors. This article describes how to avoid
these errors both for manual derivation of series and when implementing
series packages.
\end{abstract}

\section{Introduction}

This article is a companion to reference \cite{StoutemyerSeriesMisdemeanors}.
That article describes how to overcome design limitation that make
many current Puiseux-series implementations unnecessarily inconvenient,
such as not providing the order that the user requests or not allowing
requests for negative or fractional orders.

In contrast, this article describes how to overcome the more serious
problem of results that are incorrect on branch cuts and at branch
points for most current implementations. This article is relevant
to both truncated and infinite series of almost any type, including
hierarchical, Fourier, Dirichlet and Poisson series. However, for
concreteness the discussion is specific to Puiseux series that are
generalized to permit in the coefficients sub-polynomial functions
of the expansion variable, such as logarithms.

Section \ref{sec:Branch-bugs-for-ln} discusses branch bugs for logarithms.
Section \ref{sec:BugsForFracPows} discusses branch bugs for fractional
powers. Section \ref{sec:Branch-bugs-for-arc} discusses branch bugs
for inverse trigonometric and inverse hyperbolic functions. 

Let $z$ be a complex variable, and let $x$, $y$, $r$, and $\theta$
be real variables. An appropriate substitution can always transform
any expansion point including $\infty$, $-\infty$ and the complex
circle at radius $\infty$ to $z=x+iy=re^{i\theta}=0$. Therefore
without loss of generality the discussion assumes that 0 is the expansion
point.

To test an implementation for branch bugs in a truncated series $U\left(z\right)$
for an expression $u\left(z\right)$:
\begin{enumerate}
\item Evaluate $\left|u\left(0\right)-U\left(0\right)\right|$. If there
is a singularity at $z=0$, then the result might be undefined even
if the series is correct. Otherwise this absolute error should be
0 or very nearly so.\vspace{-0.05in}

\item Do a high-resolution 3D plot of $\left|u\left(z\right)-U\left(z\right)\right|\:\,|\: z\rightarrow x+iy$
centered at $z=0$ with enough terms to span an exponent range of
at least 4. Try zooming in from a moderate initial box radius. An
initial box radius of 0.5 works well for most examples in this article.\vspace{-0.05in}

\item 3D plots can easily miss discontinuities that are ribs, crevasses,
or thin cusps emanating from $z=0$ -- particularly if an edge doesn't
lie along a grid line. Therefore if step 2 doesn't reveal an incorrect
result, then\vspace{-0.05in}

\begin{enumerate}
\item Do a 2-D plot of $\left|u\left(z\right)-U\left(z\right)\right|\:|\: z\rightarrow r_{0}e^{i\theta}$
for $\theta=(-\pi,\pi]$ and various fixed $r_{0}$ that are well
within the estimated radius of convergence.
\item For each critical direction $\theta_{c}$ defined in Section \ref{sub:lnNon0DominantExponent},
plot $\left|u\left(z\right)-U\left(z\right)\right|\:|\: z\rightarrow re^{i\theta_{c}}$
for $r=[-R,R]$ and various fixed $R>0$ that are well within the
estimated radius of convergence.\vspace{-0.05in}

\end{enumerate}
\item If the result of step 1 is undefined because of a singularity, then
exclude $z=0$ or clip the plot magnitude to a positive value $\ll1$.\vspace{-0.05in}

\item Within rounding error and the radius of convergence, the surface and
the curves should converge to 0.0 as the number of terms increases.
If instead any of these plots converge to an obvious jump touching
$z=0$, then the formula is almost certainly incorrect. If there is
a hint of a jump that grows from magnitude 0.0 at $z=0$, then try
instead plotting the relative error $\left|\left(u\left(z\right)-U\left(z\right)\right)/u\left(z\right)\right|$,
excluding $z=0$.
\end{enumerate}
For a real variable $x$, it suffices instead to evaluate $u(0)-U(0)$
and to plot $\left|u\left(x\right)-U\left(x\right)\right|$ and\[
\left|\dfrac{u\left(x\right)-U\left(x\right)}{u\left(x\right)}\right|.\]

\section{Branch bugs for ln\label{sec:Branch-bugs-for-ln}}

\begin{flushright}
{}``\textsl{Spare the branch and spoil the child}.''\\
-- adapted from King Solomon's proverbs.
\par\end{flushright}

Table \ref{tab:SeriesLn} gives several logands together with one
or more correct alternatives for their dominant 0-degree terms of
the logarithm series expanded about complex $z=0$ or real $x=0$. 

\begin{table}[h]
\caption{$\mbox{0-degree term of series}\left(\ln u,\,\mbox{var}\!=\!0,\, o\left(\mathrm{var}^{n}\right)\right)$
with $n\!\geq\!4$, $x,y\in\mathbb{R}$, $z=x\!+\! iy$: \label{tab:SeriesLn}}

\begin{tabular}{|c|r|l|c|}
\hline 
\# & $u$ & Alternative 0-degree terms of $\ln u$ near  $\mbox{variable}=0$ & why\tabularnewline
\hline
\hline 
\negthinspace{}\negthinspace{}1a\negthinspace{}\negthinspace{} & $z^{2}+z^{3}$ & $\ln\!\left(z^{2}\right)+\begin{cases}
2i\pi & \mathrm{if}\:\Im\left(z\right)\!<\!0\wedge\Re\left(z\right)\geq\Im\left(z\right)^{2}\!/2+\cdots+o\left(\Im\left(z\right)^{n}\right)\\
-2i\pi & \mathrm{if}\:\Im\left(z\right)\!\geq\!0\wedge\Re\left(z\right)\geq\Im\left(z\right)^{2}\!/2+\cdots+o\left(\Im\left(z\right)^{n}\right)\\
0 & \mathrm{otherwise}\end{cases}$ & $\begin{array}{c}
(\ref{eq:OmegaIs0OrPlusOrMinus2Pi})\\
(\ref{eq:DefineAngleOfs})\\
\mathrm{\! rem}.\:\ref{rem:ExplictCurveAndSeriesThereof}\!\end{array}$\tabularnewline
\hline 
\negthinspace{}\negthinspace{}1b\negthinspace{}\negthinspace{} &  & $\!2\ln z+2i\pi\!\left\lfloor \dfrac{\pi\!-\!2\arg z}{2\pi}\right\rfloor \!+\!\begin{cases}
2i\pi, & \!\!\!\Im\left(z\right)\!<\!0\wedge\Re\left(z\right)\geq\Im\left(z\right)^{2}\!/2+\cdots\!\!\\
-2i\pi, & \!\!\!\Im\left(z\right)\!\geq\!0\wedge\Re\left(z\right)\geq\Im\left(z\right)^{2}\!/2+\cdots\!\!\\
0, & \mathrm{\!\!\! otherwise}\end{cases}$ & $\begin{array}{c}
(\ref{eq:PsiEqOmegaPlusFloorPlusConditional})\\
(\ref{eq:DefineAngleOfs})\\
\mathrm{\! rem}.\:\ref{rem:ExplictCurveAndSeriesThereof}\!\end{array}$\tabularnewline
\hline 
\negthinspace{}\negthinspace{}2a\negthinspace{}\negthinspace{} & $z^{2}+z^{3}e^{z}$ & $\ln\!\left(z^{2}\right)+\begin{cases}
2i\pi & \mathrm{if}\:\arg\!\left(1\!+\! z\!+\! z^{2}\!+\cdots+o\left(z^{n}\right)\right)+\arg\!\left(z^{2}\right)\leq-\pi\\
-2i\pi & \mathrm{if}\:\arg\!\left(1\!+\! z\!+\! z^{2}\!+\cdots+o\left(z^{n}\right)\right)+\arg\!\left(z^{2}\right)>\pi\\
0 & \mathrm{otherwise}\end{cases}$ & (\ref{eq:OmegaIs0OrPlusOrMinus2Pi})\tabularnewline
\hline 
\negthinspace{}\negthinspace{}2b\negthinspace{}\negthinspace{} &  & $\!2\ln z\!+\!2\pi i\!\left\lfloor \dfrac{\pi\!-\!2\arg z}{2\pi}\!\right\rfloor \!+\!\begin{cases}
2i\pi, & \!\mathrm{\!\!\!}\arg\!\left(1\!+\cdots+o\!\left(z^{n}\right)\right)+\arg\!\left(z^{2}\right)\leq-\pi\!\!\\
-2i\pi, & \mathrm{\!\!\!\!}\arg\!\left(1\!+\cdots+o\!\left(z^{n}\right)\right)+\arg\!\left(z^{2}\right)>\pi\\
0, & \mathrm{\!\!\! otherwise}\end{cases}$ & (\ref{eq:PsiEqFloorPlusConditional})\tabularnewline
\hline 
\negthinspace{}\negthinspace{}2c\negthinspace{}\negthinspace{} &  & $\ln\!\left(z^{2}\right)+\left(\arg\left(z^{2}\!+\cdots+o\left(z^{n}\right)\right)-\arg\!\left(1\!+\cdots+o\left(z^{n}\right)\right)-\arg\!\left(z^{2}\right)\right)i$ & (\ref{eq:OmegaUnconditional})\tabularnewline
\hline 
\negthinspace{}\negthinspace{}2d\negthinspace{}\negthinspace{} &  & $2\ln\!\left(z\right)+\left(\arg\left(z^{2}\!+\cdots+o\left(z^{n}\right)\right)-\arg\!\left(1\!+\cdots+o\left(z^{n}\right)\right)-2\arg\!\left(z\right)\right)i$ & (\ref{eq:PsiEq4Args})\tabularnewline
\hline 
3 & \negthinspace{}-$z^{\nicefrac{-7}{6}}\!-\! z^{\nicefrac{7}{3}}\!$ & $\ln\left(-z^{-7/6}\right)\:$ or $ $$\:-7\ln\left(z\right)/6+\begin{cases}
i\pi & \mathrm{if}\:\Im\left(z\right)\geq0\\
-\pi i & \mathrm{otherwise}\end{cases}$ & $\begin{array}{c}
\!\!\mathrm{prop}\:\ref{pro:omegaForExponentDivisibility}\!\!\\
\mathrm{or}\:(\ref{eq:PsiEqFloorPlusConditional})\end{array}$\tabularnewline
\hline 
4 & \negthinspace{}-$1\!-\! z^{2}\!-\! z^{3}\!$ & $\begin{cases}
i\pi & \mathrm{if}\:\left(\Im\left(z\right)\geq0\,\wedge\,\Re\left(z\right)\leq\Im\left(z\right)^{2}\!/2+\cdots+o\left(\Im\left(z\right)^{n}\right)\right)\vee\\
 & \quad\left(\Im\left(z\right)>0\,\wedge\,\Re\left(z\right)\geq\Im\left(z\right)^{2}\!/2+\cdots+o\left(\Im\left(z\right)^{n}\right)\right)\\
-i\pi & \mathrm{otherwise}\end{cases}$ & $\begin{array}{c}
(\ref{eq:OmegaPiOfImGOnC})\\
\!\mathrm{rem.}\:\ref{rem:ExplictCurveAndSeriesThereof}\end{array}$\tabularnewline
\hline 
5 & \negthinspace{}-$1\!-\! z^{2}e^{z}\!$ & $\begin{cases}
i\pi & \mathrm{if}\:\Im\left(z^{2}+2z^{3}+z^{4}+\cdots+o\left(z^{n}\right)\right)\leq0\\
-i\pi & \mathrm{otherwise}\end{cases}$ & (\ref{eq:OmegaPiOfImGOnC})\tabularnewline
\hline 
6 & $-1\!+\! iz^{1/4}$ & $i\pi$ & (\ref{eq:tauForomegaEq0})\tabularnewline
\hline 
7 & \negthinspace{}-$1\!-\! iz^{\frac{1}{4}}+z\!$ & $\begin{cases}
i\pi & \mathrm{if}\: z=0\\
-i\pi & \mathrm{otherwise}\end{cases}$ & (\ref{eq:tauForomegaEqMinus2PiExceptzEq0})\tabularnewline
\hline 
\negthinspace{}\negthinspace{}8\negthinspace{}\negthinspace{} & \negthinspace{}-$1\!+\! iz^{\frac{1}{2}}\!+\! z\!$ & $\begin{cases}
i\pi & \mathrm{if}\: z\leq0\\
-i\pi & \mathrm{otherwise}\end{cases}$ & rem.\negthinspace{} \ref{rem:BorderlineCase}\tabularnewline
\hline 
\negthinspace{}\negthinspace{}9\negthinspace{}\negthinspace{} & $c+z^{2}$ & $\begin{cases}
\ln\left(z^{2}\right) & \mathrm{if}\: c=0\\
\ln\left(c\right)-2i\pi & \mathrm{if}\:\arg\left(c\right)=\pi\wedge\Im\left(z^{2}\right)<0\\
\ln\left(c\right) & \mathrm{otherwise}\end{cases}$ & \negthinspace{}rem.\negthinspace{} \ref{rem:CasesForLiteralDominantCoefs}\negthinspace{}\tabularnewline
\hline 
\negthinspace{}\negthinspace{}10\negthinspace{}\negthinspace{} & $-x^{-2}+e^{x}$ & $\ln\left(x^{-2}\right)$ & \negthinspace{}\negthinspace{}prop\negthinspace{} \ref{pro:omegaRealForRealCfAndIntegerExpon}\negthinspace{}\negthinspace{}\tabularnewline
\hline 
\negthinspace{}\negthinspace{}11\negthinspace{}\negthinspace{} & $x^{2}+x^{3}e^{x}$ & $\ln\left(x^{2}\right),\:$ or $\:2\ln\left(\left|x\right|\right),\:$
or $\:2\ln\left(x\right)+\begin{cases}
-2i\pi & \mathrm{if}\: x<0\\
0 & \mathrm{otherwise}\end{cases}$ & $\begin{array}{c}
\!\!\mathrm{prop}\:\ref{pro:omegaRealForRealCfAndIntegerExpon}\!\!\\
\mathrm{or}\:(\ref{eq:PsiEqFloorPlusConditional})\end{array}$\tabularnewline
\hline 
\negthinspace{}\negthinspace{}12a\negthinspace{}\negthinspace{} & $x^{4/3}+x^{2}$ & $\ln\left(x^{4/3}\right)$ & \negthinspace{}rem.\negthinspace{} \ref{rem:omegaRealInGeneral}\negthinspace{}\tabularnewline
\hline 
\negthinspace{}\negthinspace{}12b\negthinspace{}\negthinspace{} & $\begin{array}{c}
\mathrm{real}\\
\mathrm{branch}\end{array}$ & $\ln\left(x^{4/3}\right),\:$ or $\:4\ln\left(|x|\right)/3,\:$ or
$\:4\ln\left(x\right)/3-\begin{cases}
4i\pi/3 & \mathrm{if}\: x<0\\
0 & \mathrm{otherwise}\end{cases}$ & $\begin{array}{c}
\mathrm{sec}.\quad\\
\ref{sub:RealBranchLnFracPow}\end{array}$\tabularnewline
\hline
\end{tabular}
\end{table}

If an implementation gives a degree-0 term that isn't equivalent to
these correct alternatives at $x=0$ and as $x\rightarrow0$ from
both directions or at $z=0$ and as $z\rightarrow0$ from all directions,
then the computer algebra result is incorrect. For example, most implementations
currently give the following generalized infinite generalized Maclaurin
series or a truncated version of it for $\ln\left(z^{2}+z^{3}\right)$:\begin{equation}
2\ln z+\sum_{k=0}^{\infty}\dfrac{(-1)^{k}z^{k+1}}{k+1}.\label{eq:AWrongSeriesForLn}\end{equation}
Within the radius of convergence 1, this series converges to values
that are too large by $2\pi i$ wherever \[
y\geq0\:\wedge\: x<\dfrac{\sqrt{1+3y^{2}}-1}{3},\]
or too small by $2\pi i$ wherever\[
y<0\:\wedge\: x\leq\dfrac{\sqrt{1+3y^{2}}-1}{3},\]
with $z=x+iy$. For example, at $z=0.1i$, $\ln\left(z^{2}+z^{3}\right)\simeq4.6002\mathbf{-3.04192}\boldsymbol{i}$,
whereas series (\ref{eq:AWrongSeriesForLn}) truncated to $o\left(z^{4}\right)$
gives approximately $4.6002\mathbf{+3.24126}\boldsymbol{i}$. Series
(\ref{eq:AWrongSeriesForLn}) is also incorrect at $z=-0.1$.

If an implementation gives a degree-0 term that is equivalent to one
of the results listed in Table \ref{tab:SeriesLn} but more complicated,
then there is room for improvement of the simplification in ways described
below.

One way to avoid returning an incorrect result is to refuse attempting
series expansions on branch cuts and on the branch points at their
ends, either returning an error indication or an unsimplified result
such as {}``$\mathrm{series}\left(\cdots\right)$''. However, this
precludes useful results for many examples of frequent interest, such
as
\begin{itemize}
\item fractional powers and logarithms of many expressions whose dominant
exponent is non-zero or whose dominant coefficient isn't positive,
\item $\arcsin\left(u\left(z\right)\right)$, $\arccos\left(u\left(z\right)\right)$,
$\mathrm{arccosh}\,\left(u\left(z\right)\right)$, and $\mathrm{arctanh}\,\left(u\left(z\right)\right)$
at $u\left(z\right)=1$ or $u\left(z\right)=-1$.
\end{itemize}
Another way to avoid returning an incorrect result is to force the
user to specify a numeric direction $\theta_{0}$ for the series expansion
variable $z=re^{i\theta}$, compute\[
\mathrm{series}\left(f\left(re^{i\theta_{0}}\right),\, r=0^{+},\, o\left(r^{n}\right)\right),\]
substitute $r\rightarrow ze^{-i\theta_{0}}$ into the result, then
preferably attach to the result the constraint {}``$|\: z=re^{i\theta_{0}}\,\wedge\, r>0$''.
The result is then guaranteed only for direction $\theta_{0}$. This
is a reasonable approach when the only purpose of the series is to
determine a uni-directional limit of an expression via that of its
dominant term, and for a bi-directional limit we can invoke $\mathrm{series}\left(\ldots\right)$
twice with two different values of $\theta_{0}$. However, this approach
isn't appropriate for omni-directional limits.

Moreover, with this approach it is important for $\theta_{0}$ to
have \textsl{no} default, such as the most likely choice 0. Otherwise
many users won't realize that their formula might not be correct for
non-positive or non-real $z$.

In contrast, this article presents formulas that are correct for all
$\theta_{0}$ that aren't precluded by any constraints provided by
the user, such as {}``$\ldots|\: z>0$'' or {}``$\ldots|\:-\pi/3<\arg z\leq2\pi/3$''.
Moreover, even for the approach of requiring a numeric $\theta_{0}$,
the formulas presented in the remainder of this article are helpful
for determining the correct behavior for that $\theta_{0}$.

\subsection{Incorrect extraction of the dominant term }

\begin{flushright}
{}``\textsl{The devil is in the details}.''\\
-- after Gustave Flaubert.
\par\end{flushright}

Most computer-algebra systems use a particular branch when a multiply-branched
function is simplified for numeric arguments. This branch is most
often the principal branch. However, some computer algebra systems
offer the option or the default of using the real branch for fractional
powers having odd reduced denominators together with real radicands.
Either way, for consistency the same branch should be used for expressions
and their series.

One source of incorrect ln series is omitting the $\Upsilon i$ term
in the following universal principal-branch formula for the distribution
of logarithms over products:\begin{eqnarray}
\ln\left(uv\right) & \equiv & \ln(u)+\ln(v)+\Upsilon i,\label{eq:LnOfProd}\end{eqnarray}
where \begin{eqnarray}
\Upsilon & = & \arg(uv)-\arg(u)-\arg(v)\label{eq:UpsilonUnconditional}\\
 & = & \begin{cases}
2\pi & \mathrm{if}\;\arg(u)+\arg(v)\leq-\pi,\\
-2\pi & \mathrm{if\;}\arg(u)+\arg(v)>\pi,\\
0 & \mathrm{otherwise}.\end{cases}\label{eq:UpsilonIs0OrPlusOrMinus2Pi}\end{eqnarray}
This can be proved from\begin{eqnarray}
\ln\left(|uv|\right) & \equiv & \ln(|u|)+\ln(|v|),\label{eq:LnAbsProd}\\
\ln\left(|w|\right) & \equiv & \ln\left(w\right)-\arg\left(w\right)i,\label{eq:LnAbs}\\
\arg(uv) & \equiv & \mathrm{mods}\left(\arg(u)+\arg(v),\,2\pi\right).\label{eq:ArgProd}\end{eqnarray}
Here $\mathrm{mod\mathbf{s}}\left(u,v\right)$ is the residue of $u$
of mod $v$ in the near-\textbf{s}ymmetric interval $(-v/2,v/2]$
for $v$ positive.%
\footnote{The name mods is inspired by that built-in Maple function.

Many of the formulas involving $\arg\left(\ldots\right)$ in this
article can be expressed more concisely using the unwinding number
described in\cite{CorlessAndJeffrey} and later redefined more conveniently
with the opposite sign in \cite{AccordingToAAndS}. But alas, unwinding
numbers aren't yet built-into the mathematics curriculum and most
computer-algebra systems.%
}
\begin{rem}
\label{rem:Arg0Eq0}These formulas require the useful but non-universal
definition\begin{eqnarray}
\arg\left(0\right) & := & 0,\end{eqnarray}
as is done in \textsl{Mathematica}$^{\lyxmathsym{\textregistered}}$.
If a built-in $\arg\left(\ldots\right)$ function does anything else,
then an implementer should prepend here and throughout this article
appropriate cases for each possible combination of an argument of
$\arg\left(\ldots\right)$ being 0. For example,\begin{eqnarray*}
\Upsilon & = & \begin{cases}
0 & \mathrm{if}\; u=0\,\vee\, v=0,\\
\arg(uv)-\arg(u)-\arg(v) & \mathrm{otherwise}.\end{cases}\\
 & = & \begin{cases}
0 & \mathrm{if}\; u=0\,\vee\, v=0\,\vee\,-\pi<\arg(u)+\arg(v)\leq\pi,\\
2\pi & \mathrm{if}\;\arg(u)+\arg(v)\leq-\pi,\\
-2\pi & \mathrm{otherwise}.\end{cases}\end{eqnarray*}
Here and throughout this article, Boolean expressions and braced case
expressions are assumed to be done using short-circuit evaluation
from left- to-right within top-to-bottom order to avoid evaluating
ill-defined sub-expressions and to avoid the clutter of making the
tests mutually exclusive.
\end{rem}
Alternative (\ref{eq:UpsilonUnconditional}) is more compact than
alternative (\ref{eq:UpsilonIs0OrPlusOrMinus2Pi}) and reveals that
jumps in $\Upsilon$ can occur only where one of $\arg(u)$, $\arg(v)$
or $\arg\left(uv\right)$ is $\pi$. However, alternative (\ref{eq:UpsilonIs0OrPlusOrMinus2Pi})
is more candid because it makes the piecewise constancy manifest rather
than cryptic. Moreover, approximate values are often substituted into
expressions for purposes such as plotting, and the conditional alternative
(\ref{eq:UpsilonIs0OrPlusOrMinus2Pi}) avoids having the magnitude
of the imaginary part of a result be several machine $\varepsilon$
when it should be 0: Unlike the unconditional alternative, the conditional
alternative never subtracts two approximate angles from approximately
$\pi$, giving approximately 0.

For $\mathrm{series}\left(\ln\left(\ldots\right),\, z=0,\, o\left(z^{n}\right)\right)$
with negative $n$, the result is $0+o\left(z^{n}\right)$ if the
logand doesn't contain an essential singularity. In contrast, for
a non-negative requested $n$, the usual algorithm for computing the
logarithm of a series entails converting the dominant term of the
logand series to 1 by factoring out the dominant term then distributing
the logarithm over the resulting product:\begin{eqnarray}
\ln\left(c\left(z\right)z^{\alpha}+g(z)\right) & \rightarrow & \ln\left(c\left(z\right)z^{\alpha}\left(1+\frac{g(z)}{c\left(z\right)z^{\alpha}}\right)\right)\nonumber \\
 & \rightarrow & \left(\Omega i+\ln\left(c(z)z^{\alpha}\right)\right)+\ln\left(1+\frac{g(z)}{c\left(z\right)z^{\alpha}}\right).\label{eq:LnOfASeries}\end{eqnarray}
Here $c\left(z\right)z^{\alpha}$ is the dominant term and $g(z)$
is the sum of all the other terms, with alternatives (\ref{eq:UpsilonUnconditional})
and (\ref{eq:UpsilonIs0OrPlusOrMinus2Pi}) giving\begin{eqnarray}
\Omega & = & \arg\left(c(z)z^{\alpha}+g(z)\right)-\arg\left(1+\frac{g(z)}{c\left(z\right)z^{\alpha}}\right)-\arg\left(c\left(z\right)z^{\alpha}\right)\label{eq:OmegaUnconditional}\\
 & = & \begin{cases}
2\pi & \mathrm{if}\:\arg\left(1+\frac{g(z)}{c\left(z\right)z^{\alpha}}\right)+\arg\left(c\left(z\right)z^{\alpha}\right)\leq-\pi,\\
-2\pi & \mathrm{if}\:\arg\left(1+\frac{g(z)}{c\left(z\right)z^{\alpha}}\right)+\arg\left(c\left(z\right)z^{\alpha}\right)>\pi,\\
0 & \mathrm{otherwise}.\end{cases}\label{eq:OmegaIs0OrPlusOrMinus2Pi}\end{eqnarray}

Always $\Omega=0$ at $z=0$.

It is important to simplify $\Omega$ as much as is practical for
each particular logand series. Not only is the result more intelligible
-- it is usually also more accurate for approximate computation: There
can be catastrophic cancellation between terms of $g(z)$. Therefore
without good algebraic simplification, $\Omega$ can be dramatically
incorrect along and near branch cuts when evaluated with approximate
arithmetic.
\begin{rem}
\label{rem:OmegaEq0IfGEq0}For example, $\Omega\equiv0$ if $g(z)\equiv0$,
because then $\arg\left(1+g(z)/\left(c(z)z^{\alpha}\right)\right)\equiv0$
and $\arg\left(c\left(z\right)z^{\alpha}\right)$ must be in the interval
$(-\pi,\pi]$.
\end{rem}
Even if we can't determine a simpler formula that is equivalent to
formula (\ref{eq:OmegaIs0OrPlusOrMinus2Pi}) throughout the entire
complex plane or the entire real line for real $z$, $\Omega$ might
be equivalent to a simpler expression $\omega$ throughout the radius
of convergence or the useful portion thereof. It is especially important
to exploit this for examples such as\[
\textrm{series}\left(z^{-1}+e^{z}\ln\left(-2-z\right),\, z\!=\!0,\, o\left(z^{5}\right)\right)\]
where $\Omega$ or $\omega$ infects all but one result term, giving
a bulky result that is difficult to comprehend.

Let $R>0$ be the classic radius of convergence computed disregarding
any closer branch cuts, or let $R$ be the {}``radius of computational
utility'' for divergent series. Let $\underline{R}>0$ be the largest
radius from $z=0$ within which $\Omega$ and a simpler $\omega$
give identical values. Radius $\underline{R}$ can be an arbitrarily
small portion of $R$ because a branch cut can pass arbitrarily close
to $z=0$. However, it is almost always justifiable to use $\omega$
in place of $\Omega$ because:\vspace{-0.1in}

\begin{itemize}
\item If our purpose is to determine the local behavior of the series at
$z=0$, such as for computing a limit by computing the limit of the
dominant term, then any $\underline{R}>0$ is sufficient justification
for using $\omega$.\vspace{-0.1in}

\item It seems pointless to use $\Omega$ rather than a simpler $\omega$
for any purpose if $\underline{R}>R$ or if $\underline{R}$ is greater
than the percentage of $R$ beyond which convergence is impractically
slow or subject to unacceptable catastrophic cancellation. \vspace{-0.1in}

\item If we don't expect a generalized Puiseux series to capture infinite
magnitudes associated with singularities not at $z=0$, then why should
we expect such series to capture the less severe finite-magnitude
jumps associated with branch cuts that don't touch $z=0$?\vspace{-0.1in}

\item We can take the view that the \textsl{generalized} radius of convergence
is the distance to the nearest singularity \textsl{or jump} in $\Omega$
that we can't account for with $\omega$, and there should be no expectation
that a series is truthful beyond its generalized radius of convergence.
\end{itemize}
Here is one such opportunity for computing an $\omega$ that is significantly
simpler than $\Omega$:
\begin{prop}
\label{pro:omegaRealForRealCfAndIntegerExpon}Let $\omega_{\mathrm{real}}$
denote $\omega$ for the special case of real $z$. If all of the
terms in the truncated logand series have real coefficients and integer
exponents, then $\omega_{\mathrm{real}}\equiv0$.\end{prop}
\begin{proof}
If the truncated logand series has all integer powers of real $z$
and all real coefficients, then $1+g\left(z\right)/\left(c\left(z\right)z^{\alpha}\right)$
is real for all real $z$, and $\arg\left(1+g\left(z\right)/\left(c\left(z\right)z^{\alpha}\right)\right)=0$
for all real $z$ such that $g\left(z\right)/\left(c\left(z\right)z^{\alpha}\right)\geq-1$.
There is a singularity wherever $g(z)/\left(c(z)z^{\alpha}\right)=-1$,
providing an upper bound on the radius of convergence. Also, $-\pi<\arg\left(c(z)z^{\alpha}\right)\leq\pi$.
Therefore throughout the radius of convergence\[
-\pi<\arg\left(1+\frac{g(z)}{c\left(z\right)z^{\alpha}}\right)+\arg\left(c\left(z\right)z^{\alpha}\right)\leq\pi\]
in equation (\ref{eq:OmegaIs0OrPlusOrMinus2Pi}), making $\omega_{\mathrm{real}}\equiv0$.\end{proof}
\begin{rem}
\label{rem:omegaRealForOddDenominators}If we are in a mode that consistently
uses the real branch for fractional powers having odd denominators,
such as the TI-Nspire%
\footnote{The computer algebra used in Texas Instruments products has no name
separate from the variously named calculators, Widows and Macintosh
products that contain it. The most recent such product is named TI-Nspire.%
} real mode, then more generally $\omega_{\mathrm{real}}\equiv0$ if
all of the coefficients are real and none of the reduced exponents
have even denominators.
\end{rem}

\subsubsection{The non-zero dominant exponent case\label{sub:lnNon0DominantExponent}}

\begin{flushright}
{}``\textsl{Beware of geeks bearing formulas}.''\\
-- Warren Buffett.
\par\end{flushright}

Here is a useful easy simplification test for real or complex $z$
when the dominant exponent $\alpha\neq0$:
\begin{prop}
\label{pro:omegaForExponentDivisibility}If all of the coefficients
in a truncated logand series $U$ are real and all of the exponents
of $z$ in $U$ are integer multiples of a non-zero dominant exponent,
then $\omega\equiv0$.\end{prop}
\begin{proof}
Let the truncated logand series be $U=c\left(z\right)z^{\alpha}+g(z)$
with dominant term $c\left(z\right)z^{\alpha}$ having non-zero $\alpha$.
Also, let {}``\textbf{near} $z=0$'' denote all\begin{eqnarray}
z & = & re^{i\theta}\:\:|\:\: r\!>\!0\,\wedge\,\left|\dfrac{g(z)}{c(z)z^{\alpha}}\right|\!<\!1,\end{eqnarray}
which bounds the radius of convergence. From equation (\ref{eq:OmegaUnconditional}),
it is apparent that $\Omega$ can jump only where one of its three
$\arg\left(\ldots\right)$ terms jumps. Term $\arg\left(1+g(z)/\left(c(z)z^{\alpha}\right)\right)$
can't jump near $z=0$ because $c(z)z^{\alpha}$ dominates $g(z)$.
Every term in $c(z)z^{\alpha}+g(z)$ can be expressed as $k(z)\left(z^{\alpha}\right)^{m}$
with real $k(z)$ and integer $m$, because the transformation $z^{\gamma}\rightarrow\left(z^{\alpha}\right)^{m}$
with $m=\gamma/\alpha$ is always valid for integer $m$, as discussed
in Section \ref{sub:Principal-branch-series-of-frac-pow}. Thus $c\left(z\right)z^{\alpha}+g(z)$
is real where\begin{eqnarray}
\arg\left(c\left(z\right)z^{\alpha}\right) & = & \pi,\label{eq:ArgczToAlphaEqPi}\end{eqnarray}
which are the only places that $\arg\left(c\left(z\right)z^{\alpha}\right)$
jumps. Since $c\left(z\right)z^{\alpha}$ is negative there and dominates
$g(z)$, adding the relatively small-magnitude real $g(z)$ to negative
$c\left(z\right)z^{\alpha}$ leaves\begin{eqnarray}
\arg\left(c\left(z\right)z^{\alpha}+g(z)\right) & = & \pi\label{eq:ArgSeriesEqPi}\end{eqnarray}
near $z=0$. In equation (\ref{eq:OmegaUnconditional}), $\arg\left(c\left(z\right)z^{\alpha}+g(z)\right)$
and $\arg\left(c\left(z\right)z^{\alpha}\right)$ have opposite signs.
Therefore these jumps cancel within the radius of convergence.
\end{proof}
Propositions  \ref{pro:omegaRealForRealCfAndIntegerExpon} and \ref{pro:omegaForExponentDivisibility}
don't identify all opportunities for dramatically simplifying $\omega$.
Consider equation (\ref{eq:OmegaUnconditional}) in the neighborhood
of $z=0$. Let\begin{eqnarray}
\eta\left(\theta\right) & = & \lim_{r\rightarrow0^{+}}\arg\left(c\left(re^{i\theta}\right)\right).\label{eq:DefineEta}\end{eqnarray}

Then in the punctured neighborhood of $z=0$, the term $\arg\left(c\left(z\right)z^{\alpha}\right)$
is\begin{eqnarray}
\lim_{r\rightarrow0^{+}}\arg\left(c\left(re^{i\theta}\right)\left(re^{i\theta}\right)^{\alpha}\right) & = & \mathrm{mods}\left(\eta\left(\theta\right)+\alpha\theta,2\pi\right).\end{eqnarray}

Let $\hat{\eta}$ denote $\eta\left(\theta\right)$ in formula (\ref{eq:DefineEta})
in the common case when $\eta\left(\theta\right)$ is independent
of $\theta$ for all $-\pi<\theta\leq\pi$ not precluded by constraints
provided by the user or introduced by the computer algebra system.
Solving $\mathrm{mods}\left(\hat{\eta}+\alpha\theta_{c},2\pi\right)=\pi$
for the critical angles $\theta_{c}$ gives\begin{equation}
-\pi\,<\,\theta_{c}=\dfrac{\left(2n+1\right)\pi-\hat{\eta}}{\alpha}\,\leq\,\pi,\label{eq:SolutionForThetaC}\end{equation}
with integer $n$. Solving the inequalities in (\ref{eq:SolutionForThetaC})
for $n$ gives for $\alpha>0$ \begin{equation}
\left\lfloor \dfrac{\hat{\eta}}{2\pi}-\dfrac{1}{2}-\dfrac{\alpha}{2}\right\rfloor <n\leq\left\lfloor \dfrac{\hat{\eta}}{2\pi}-\dfrac{1}{2}+\dfrac{\alpha}{2}\right\rfloor ,\end{equation}
\emph{versus} for $\alpha<0$\begin{equation}
\left\lceil \dfrac{\hat{\eta}}{2\pi}-\dfrac{1}{2}+\dfrac{\alpha}{2}\right\rceil \leq n<\left\lceil \dfrac{\hat{\eta}}{2\pi}-\dfrac{1}{2}-\dfrac{\alpha}{2}\right\rceil .\end{equation}

There are no such angles if $\left|\alpha\right|<1$ and $\hat{\eta}$
is sufficiently close to 0, in which case $\omega\equiv0$. More generally
let $\underline{\eta}$ be a lower bound and $\bar{\eta}$ be an upper
bound on $\eta\left(\theta\right)$ over $-\pi<\theta\leq\pi$ not
excluded by any constraints. Then there are \textsl{no} critical angles
if \begin{equation}
\left(\left|\alpha\right|-1\right)\pi<\underline{\eta}\;\wedge\;\overline{\eta}<\left(1-\left|\alpha\right|\right)\pi.\label{eq:etaForomegaEq0}\end{equation}

Here is how we can proceed when none of the above tests are beneficial:
As $z\rightarrow0$,\begin{equation}
c(z)z^{\alpha}+g(z)\,=\, c(z)z^{\alpha}\left(1+\dfrac{g(z)}{c(z)z^{\alpha}}\right)\,\rightarrow\, c(z)z^{\alpha}.\end{equation}
Therefore the solution curves to equations (\ref{eq:ArgczToAlphaEqPi})
and (\ref{eq:ArgSeriesEqPi}) pair to form cusps where they don't
coincide to cancel. Along a critical angle $\theta_{c}$, $\arg\left(c(z)z^{\alpha}\right)=\pi$.
If also $\arg\left(1+g(z)/\left(c(z)z^{\alpha}\right)\right)=0$ along
$\theta_{c}$ near $z=0$, then $\theta=\theta_{c}$ is also the companion
solution to equation (\ref{eq:ArgSeriesEqPi}) near $z=0$. The $\arg\left(\ldots\right)$
terms for equations (\ref{eq:ArgczToAlphaEqPi}) and (\ref{eq:ArgSeriesEqPi})
have opposite signs in formula (\ref{eq:OmegaUnconditional}), so
these jumps cancel within the radius of convergence. Since $\arg\left(1+g(z)/\left(c(z)z^{\alpha}\right)\right)$
can't jump near $z=0$, $\omega$ is then identically 0 in the angular
neighborhood of $\theta_{c}$ near $z=0$.

If $\omega=0$ in the angular neighborhood of every critical angle,
then $\omega\equiv0$. For example, this is true for $\ln\left(-z^{2}+z^{3}\right)$
and $\ln\left(-z^{-2}+z\right)$, which don't satisfy Proposition
\ref{pro:omegaForExponentDivisibility} or inequality (\ref{eq:etaForomegaEq0}).
Even if there are some non-zero cusps, constraints might exclude those
cusps for a non-infinitesimal distance from $z=0$. Even if there
are some included cusps, it might be possible to omit either the $2\pi$
case or the $-2\pi$ case as follows: 

For $\alpha>0$, as $\theta$ increases through $\theta_{c}$, $\arg\left(c\cdot\!\left(re^{i\theta}\right)^{\alpha}\right)$
increases with $\theta$ on both sides of a jump down from $\pi$
to $(-\pi)^{+}$. If $\arg\left(1+g(z)/\left(cz^{\alpha}\right)\right)>0$
along $\theta_{c}$ near $z=0$, then from equation (\ref{eq:OmegaIs0OrPlusOrMinus2Pi}),
$\omega=-2\pi$ along and clockwise of $\theta_{c}$ until but excluding
the curved solution to equation (\ref{eq:ArgSeriesEqPi}). If instead
$\arg\left(1+g(z)/\left(cz^{\alpha}\right)\right)<0$ along $\theta_{c}$
near $z=0$, then $\omega=2\pi$ counter-clockwise of $\theta_{c}$
through the curved solution to equation (\ref{eq:ArgSeriesEqPi}).

Similarly for $\alpha<0$, for which $\arg\left(c\cdot\!\left(re^{i\theta}\right)^{\alpha}\right)$
decreases on both sides of a jump up from $(-\pi)^{+}$ to $\pi$
as $\theta$ increases through $\theta_{c}$: If $\arg\left(1+g(z)/\left(cz^{\alpha}\right)\right)>0$
along $\theta_{c}$ near $z=0$, then $\omega=-2\pi$ along and counter-clockwise
of $\theta_{c}$ until but excluding the curved solution to equation
(\ref{eq:ArgSeriesEqPi}). If instead $\arg\left(1+g(z)/\left(cz^{\alpha}\right)\right)<0$
along $\theta_{c}$ near $z=0$, then $\omega=2\pi$ clockwise of
$\theta_{c}$ through the curved solution to equation (\ref{eq:ArgSeriesEqPi}).

Thus we can omit the $2\pi$ case if for all critical directions $\arg\left(1+g(z)/\left(cz^{\alpha}\right)\right)\geq0$,
or we can omit the $-2\pi$ case if for all critical directions $\arg\left(1+g(z)/\left(cz^{\alpha}\right)\right)<0$.
For example, we can omit the $2\pi$ case for $\ln\left(iz^{2}+iz^{4}\right)$
and we can omit the $-2\pi$ case for $\ln\left(z+iz^{2}\right)$.

Notice that although $\omega$ can be $2\pi$ either clockwise or
counter-clockwise of a critical direction, $\omega$ can't be $2\pi$
\textsl{along} a critical direction.

For real $t$ let\begin{eqnarray*}
\mathrm{signum}(t) & := & \begin{cases}
-1, & \mathrm{if}\; t<0,\\
0, & \mathrm{if}\; t=0,\\
1, & \mathrm{if}\; t>0.\end{cases}\end{eqnarray*}
To determine whether expression $\arg\left(1+g(z)/\left(cz^{\alpha}\right)\right)$
is zero, positive or negative along $\theta_{c}$ near (but not at)
$z=0$, we could attempt computing\[
\underset{r\rightarrow0^{+}}{\lim}\,\mathrm{signum}\left(\arg\left(1+\dfrac{g\left(re^{i\theta_{c}}\right)}{c\cdot\!\left(re^{i\theta_{c}}\right)^{\alpha}}\right)\right).\]
However, such limits are beyond the capabilities of most computer
algebra systems if $g$ has more than one term -- particularly if
any of the exponents are fractional and/or any of the coefficients
depend on $z$. A simpler surrogate within the radius of convergence
is to instead compute\begin{equation}
\underset{r\rightarrow0^{+}}{\lim}\mathrm{\, signum}\left(\Im\left(\dfrac{g\left(re^{i\theta_{c}}\right)}{c\cdot\!\left(re^{i\theta_{c}}\right)^{\alpha}}\right)\right).\label{eq:limSignumIm}\end{equation}
However, this limit is also often beyond the capabilities of most
computer algebra systems. Fortunately there is an easily-computed
surrogate for formula (\ref{eq:limSignumIm}): Let $s\left(re^{i\theta}\right)$
be the coefficient and $\beta$ be the exponent of the lowest-degree
term of $g(z)/\left(c(z)z^{\alpha}\right)$ for which\begin{eqnarray}
\left(I_{c}\left(\dfrac{g(z)}{cz^{\alpha}},\theta_{c}\right):=\underset{r\rightarrow0^{+}}{\lim}\Im\left(s(re^{i\theta_{c}})e^{i\beta\theta_{c}}\right)\right) & \neq & 0,\label{eq:DefineAngleOfs}\end{eqnarray}
if any such term exists. Let $I_{c}\left(\ldots,\theta_{c}\right):=0$
if no such $s\left(re^{i\theta}\right)$ exists. The imaginary part
of $s\left(re^{i\theta_{c}}\right)\left(re^{i\theta_{c}}\right)^{\beta}$
dominates the imaginary parts of all subsequent terms of $g(z)$ along
$\theta_{c}$ if any exist. Therefore $I_{c}$ is a more-easily computed
surrogate for determining whether $\arg\left(1+g(z)/\left(cz^{\alpha}\right)\right)$
is zero, positive or negative along $\theta_{c}$ near $z=0$. The
term limits for computing $I_{c}$ are trivial in the common case
when a coefficient is a numeric constant, and easy even if a coefficient
is a typical sub-polynomial function of $z$. We can use $\Omega$
given by formula (\ref{eq:OmegaUnconditional}) or (\ref{eq:OmegaIs0OrPlusOrMinus2Pi})
if a coefficient contains an indeterminant other than $z$ that isn't
sufficiently constrained to decide the sign of the imaginary part
and if no imaginary parts were non-zero for previous terms.

The order to which the logand series is computed might not reveal
the lowest-degree term for which one of the $I_{c}\left(\ldots,\theta_{c}\right)$
isn't 0, thus affecting the resulting value of $\Omega$. For example,
if the series for $\ln\left(z^{2}-iz^{3}-iz^{6}\right)$ is computed
to $o\left(z^{n}\right)$ with $n\geq4$, then the dominant term is\[
\ln\left(z^{2}\right)+\begin{cases}
2\pi & \arg\left(1-iz-iz^{4}+\cdots+o\left(z^{n}\right)\right)+\arg\left(z^{2}\right)\leq-\pi,\\
0 & \mathrm{otherwise}.\end{cases}\]
However, the piecewise term is absent if the series is computed to
$o\left(z^{3}\right)$, making it incorrect by about $2\pi i$ for
values such as $z=i/10-10^{-6}$ and $z=-i/10+10^{-6}$. It is disturbing
that computing additional terms of a logand can thus change the zero-degree
term of the logarithm series. Moreover, it is awkward for algorithms
that compute additional terms incrementally, such as described by
Norman \cite{Norman}. Such is the nature of partial information along
or near a branch cut.

A way to avoid this annoyance is to use the original logand expression
in the $\arg\left(\ldots\right)$ sub-expressions of alternative (\ref{eq:OmegaUnconditional})
or (\ref{eq:OmegaIs0OrPlusOrMinus2Pi}) rather than using a truncated
series for that logand. However, the original logand isn't always
available: Perhaps as the logand we are given a truncated series,
not knowing a closed-form expression that it approximates. Even if
we did know, using the original expression can make computing a series
in one step give a different result than series composition. Such
compositional inconsistency is an undesirable property.

More seriously, including the original non-series expression as a
proper sub-expression of a series result is worse than simply returning
the original expression rather than a series: Presumably the user
requested the series to obtain insight about the behavior of the function,
or to enable symbolic operations that otherwise couldn't be done,
or to enable a numeric approximation. Returning a result that contains
the original expression as a proper sub-expression thwarts all of
these objectives. The annoyance of order-dependent coefficients is
vastly preferable. A consolation is that a zero-degree term that differs
so dramatically from that of the infinite series is at least appropriate
for a nearby problem across the branch cut. 
\begin{rem}
\label{rem:ExplictCurveAndSeriesThereof}If the truncated logand series
is simple enough so that for remaining cusps we can solve\begin{eqnarray}
\arg\left(c\left(x+iy\right)\left(x+iy\right)^{\alpha}+g\left(x+iy\right)\right) & = & \pi\label{eq:ArgSeriesRectEqPi}\end{eqnarray}
for $x\left(y\right)$ or for $y\left(x\right)$, or else solve\begin{eqnarray}
\arg\left(c\left(re^{i\theta}\right)\left(re^{i\theta}\right)^{\alpha}+g\left(re^{i\theta}\right)\right) & = & \pi\label{eq:ArgSeriesPolarEqPi}\end{eqnarray}
for $\theta\left(r\right)$ or for $r\left(\theta\right)$, then we
can construct a more candidly explicit representation of $\omega$.
For example, with $\ln\left(z^{2}+z^{3}\right)$,\begin{eqnarray}
\Omega & = & \begin{cases}
2\pi & \mathrm{if}\:\,\Im\left(z\right)<0\,\wedge\,0<\Re\left(z\right)\leq\dfrac{\sqrt{1+3\,\Im\left(z\right)^{2}}-1}{3},\\
-2\pi & \mathrm{if}\:\,\Im\left(z\right)>0\,\wedge\,0\leq\Re\left(z\right)<\dfrac{\sqrt{1+3\,\Im\left(z\right)^{2}}-1}{3},\\
0 & \mathrm{otherwise}.\end{cases}\label{eq:ExplicitCurve}\end{eqnarray}
There is catastrophic cancellation here computing $\sqrt{1+3\Im\left(z\right)^{2}}-1$
for $\left|\Im\left(z\right)\right|\ll1$. This can be avoided by
expanding the expression into the series $\Im\left(z\right)^{2}/2-3\Im\left(z\right)^{4}/8+\ldots$,
which is also more consistent with the user's request for a \textsl{series}
result. Perhaps generalized series reversion could be used to obtain
an explicit truncated series solution to equation (\ref{eq:ArgSeriesRectEqPi})
or (\ref{eq:ArgSeriesPolarEqPi}) even when we can't solve them exactly
-- at least when the coefficients are all numeric and the exponents
are all non-negative integers.
\end{rem}

\begin{rem}
\label{rem:omegaRealInGeneral}When $z$ is real, we can almost always
simplify $\omega_{\mathrm{real}}$ to either 0 or a two-piece result
that is $-2\pi$ for negative $z$ and/or positive $z$, but 0 everywhere
else as follows:

If 0 isn't a critical angle or $I_{c}(g(z)/\left(cz^{\alpha}\right),0)\leq0$,
then $\omega_{\mathrm{real}}=0$ for $z\geq0$. Otherwise $\omega_{\mathrm{real}}=-2\pi$
for $z>0$.

If $\pi$ isn't a critical angle or $I_{c}(g(z)/\left(cz^{\alpha}\right),\pi)\leq0$,
then $\omega_{\mathrm{real}}=0$ for $z\leq0$. Otherwise $\omega_{\mathrm{real}}=-2\pi$
for $z<0$.
\end{rem}

\subsubsection{The degree-0 case\label{sub:The-degree-0-case}}

If $\alpha=0$ and $\overline{\eta}<\pi$, then $\omega\equiv0$ because
for formula (\ref{eq:OmegaIs0OrPlusOrMinus2Pi}), $\arg\left(1+g(z)/\left(c(z)z^{\alpha}\right)\right)\rightarrow0$.
If instead $\alpha=0$ and $\arg\left(c\left(z\right)\right)\equiv\pi$
throughout all $z$ near $z=0$ that aren't excluded by any constraints,
then $\arg\left(1+g(z)/\left(c(z)z^{\alpha}\right)\right)+\arg\left(c\left(z\right)\right)$
can't be less than $-\pi$, because $\arg\left(1+g(z)/\left(c(z)z^{\alpha}\right)\right)\rightarrow0$.
Therefore we can at least omit the $2\pi$ case. Thus using $\omega_{\pi,0}$
to denote $\omega$ for this special case of $z=0$ being \textsl{in}
the branch cut,\begin{eqnarray}
\omega_{\pi,0} & = & \begin{cases}
-2\pi & \mathrm{if}\:\arg\left(1+\dfrac{g(z)}{c(z)}\right)>0,\\
0 & \mathrm{otherwise}.\end{cases}\label{eq:DefineOmegaSubPi}\end{eqnarray}

There are singularities wherever $g(z_{s})/c(z_{s})=-1$, and those
singularities aren't at the expansion point $z=0$. Consequently the
radius of convergence doesn't exceed the least of those $\left|z_{s}\right|$.
Therefore\begin{eqnarray}
\omega_{\pi,0} & = & \begin{cases}
-2\pi & \mathrm{if}\:0<\arg\left(\dfrac{g(z)}{c(z)}\right)<\pi,\\
0 & \mathrm{otherwise}.\end{cases}\label{eq:OmegaPiOfArgGOnC}\end{eqnarray}
Formula (\ref{eq:OmegaPiOfArgGOnC}) avoids adding a small magnitude
quantity to 1.0 near $z=0$, so this formula is more likely to be
correct than formula (\ref{eq:DefineOmegaSubPi}) with approximate
arithmetic.

Whenever we are comparing $\arg\left(f\left(z\right)\right)$ with
$-\pi/2$, 0, $\pi/2$, or $\pi$, we can often simplify the test
further and make it more accurate for approximate numbers by writing
the comparison in terms of $\Re\left(f\left(z\right)\right)$ and/or
$\Im\left(f\left(z\right)\right)$ -- particularly when $f\left(z\right)$
has more than one term. Thus\begin{eqnarray}
\omega_{\pi,0} & = & \begin{cases}
-2\pi & \mathrm{if}\:\Im\left(\dfrac{g(z)}{c(z)}\right)>0,\\
0 & \mathrm{otherwise}.\end{cases}\label{eq:OmegaPiOfImGOnC}\end{eqnarray}

This test can often be simplified further: Let $b(z)z^{\sigma}$ be
the dominant term of $g(z)/c(z)$, and let\begin{eqnarray}
\tau(\theta) & = & \lim_{r\rightarrow0^{+}}\arg\left(b\left(re^{i\theta}\right)\right).\end{eqnarray}
Let $\hat{\tau}$ denote $\tau$ in the common case where it is independent
of $\theta$ for all $0<\theta\leq\pi$ that aren't excluded by any
constraints.

In the neighborhood of $z=0$, as $\theta$ increases from $\left(-\pi\right)^{+}$
through $\pi$, $\arg\left(g(z)/\left(c(z)z^{\alpha}\right)\right)$,
$\arg\left(b(z)z^{\sigma}\right)$ and $\arg\left(\hat{\tau}z^{\alpha}\right)$
increase from (-$\sigma\pi)^{+}+\hat{\tau}$ through $\sigma\pi+\hat{\tau}$,
but jumping down from $\pi$ to $\left(-\pi\right)^{+}$ as $\theta$
increases past every critical angle where $\mathrm{\mathrm{mods}}(\sigma\theta_{c}+\hat{\tau},2\pi)=\pi$.
Expression $\Im\left(g(z)/c(z)\right)$ changes sign at those places
and also at every critical angle where $\mathrm{\mathrm{mods}}(\sigma\theta_{c}+\hat{\tau},2\pi)=0$.
There are no critical angles of either type if $0<\sigma<\nicefrac{1}{2}$
and $\tau$ is sufficiently close to $\pi/2$ or $-\pi/2$. More specifically
if $\underline{\tau}$ is a lower bound on $\tau(\theta)$ and $\overline{\tau}$
is an upper bound, then $\omega_{\pi,0}\equiv0$ when\begin{equation}
\left(\sigma-1\right)\pi<\underline{\tau}\:\wedge\:\overline{\tau}<-\sigma\pi.\label{eq:tauForomegaEq0}\end{equation}
If instead\begin{equation}
\sigma\pi<\underline{\tau}\:\wedge\:\overline{\tau}<\left(1-\sigma\right)\pi,\label{eq:tauForomegaEqMinus2PiExceptzEq0}\end{equation}
 then\label{rem:omegaRealInGeneral-1}\begin{eqnarray}
\omega_{\pi,0} & = & \begin{cases}
0 & \mathrm{if}\: z=0,\\
-2\pi & \mathrm{otherwise}.\end{cases}\label{eq:omegaPiIs0OnlyAtzEq0}\end{eqnarray}
As examples, $\omega_{\pi,0}\equiv0$ for $\ln\left(-1+iz^{1/4}\right)$,
whereas equation (\ref{eq:omegaPiIs0OnlyAtzEq0}) applies to $\ln\left(-1-iz^{1/4}\right)$.
\begin{rem}
\label{rem:BorderlineCase}If critical angles occur only at the edges
of regions not excluded by constraints, then we can compute $I_{c}\left(g(z)/c(z),\ldots\right)$
at those critical angles and at one included non-critical angle to
attempt simplifying formula (\ref{eq:OmegaPiOfImGOnC}). For example
with $\ln\left(-1+iz^{1/2}+z\right)$, the one critical angle is $\pi$,
along which $I_{s}$ is non-positive. $I_{c}$ is also non-positive
along the included non-critical angle $-\pi/2$, so $\omega_{\pi,0}\equiv0$.
\end{rem}
In contrast for $\ln\left(-1-iz^{1/2}+z\right)$, the one critical
angle is $\pi$, along which $I_{c}$ is non-positive. However, $I_{c}$
is positive along the included non-critical angle $\pi/2$. Therefore
$\omega$ is 0 for the ray $z\leq0$ but $-2\pi$ everywhere else.

If a critical angle is interior to an included region, then we can
simply use formula (\ref{eq:OmegaPiOfImGOnC}). Remark \ref{rem:omegaRealInGeneral}
is also applicable to $\omega_{\pi,0}$ for real $z$.
\begin{rem}
\label{rem:CasesForLiteralDominantCoefs}For an example such as $\ln\left(c+z^{2}+z^{3}\right)$
where a value such as 1, -1 or 0 could subsequently be substituted
for the literal constant $c$, a completely correct result should
piecewise account for all possible cases. 
\end{rem}

\subsection{Incorrect extraction of the dominant coefficient\label{sub:ExtractionOfDomCoef}}

A logarithm of a product is more concise and more efficient to approximate
numerically than a sum of logarithms of the factors. For this reason,
users often prefer a result that contains a logarithm of a product
rather than an equivalent sum of logarithms.

However, at this time many existing generalized Puiseux-series implementations
always distribute the logarithm of the dominant term over its coefficient
and cofactor. This distribution is justified for a hierarchical series
if the coefficient is itself a series in logarithms or nested logarithms
depending on $z$ -- at least when this distribution unifies two otherwise
different logarithms. For example, it is more consistent to have $\ln\left(z\right)$
throughout a series than to have $\ln\left(z\right)$ in some places
and $\ln\left(-2z\right)$ in other places.

However, this distribution provides an additional opportunity for
incorrect results caused by ignoring the $\Upsilon i$ term in equation
(\ref{eq:LnOfProd}). Applying this transformation twice and simplifying
gives\begin{eqnarray}
\ln\left(c\left(z\right)z^{\alpha}+g(z)\right) & \rightarrow & \left(\Phi i+\ln\left(c\left(z\right)\right)+\ln\left(z^{\alpha}\right)\right)+\ln\left(1+\frac{g(z)}{c\left(z\right)z^{\alpha}}\right),\end{eqnarray}
where\begin{eqnarray}
\Phi & = & \Omega+\arg\left(c(z)z^{\alpha}\right)-\arg\left(c(z)\right)-\arg\left(z^{\alpha}\right),\label{eq:PhiEqOmegaPlusUnconditional}\\
 & = & \Omega+\begin{cases}
2\pi & \mathrm{if}\:\arg\left(c(z)\right)+\arg\left(z^{\alpha}\right)\leq-\pi,\\
-2\pi & \mathrm{if}\:\arg\left(c(z)\right)+\arg\left(z^{\alpha}\right)>\pi,\\
0 & \mathrm{otherwise},\end{cases}\label{eq:PhiEqOmegaPlusConditional}\\
 & = & \arg\left(c\left(z\right)z^{\alpha}+g(z)\right)-\arg\left(1+\frac{g(z)}{c\left(z\right)z^{\alpha}}\right)-\arg\left(c(z)\right)-\arg\left(z^{\alpha}\right),\label{eq:PhiEq4Args}\\
 & = & \begin{cases}
2\pi & \mathrm{if}\:\arg\left(c(z)\right)+\arg\left(z^{\alpha}\right)+\arg\left(1+\dfrac{g(z)}{c(z)z^{\alpha}}\right)\leq-\pi,\\
-2\pi & \mathrm{if}\:\arg\left(c(z)\right)+\arg\left(z^{\alpha}\right)+\arg\left(1\!+\!\dfrac{g(z)}{c(z)z^{\alpha}}\right)>\pi,\\
0 & \mathrm{otherwise}.\end{cases}\label{eq:PhiIs2PiMinus2PiOr0}\end{eqnarray}
Alternatives (\ref{eq:PhiEqOmegaPlusUnconditional}) and (\ref{eq:PhiEqOmegaPlusConditional})
are advantageous when $\omega$ is a constant, because only one term
occurs in the arguments of $\arg\left(\ldots\right)$. Otherwise alternative
(\ref{eq:PhiIs2PiMinus2PiOr0}) is more concise and candid.

Let $\underline{\upsilon}$ be a lower bound and $\overline{\upsilon}$
be an upper bound on $\mathrm{mods}\left(\alpha\theta,2\pi\right)$
for all angles $-\pi<\theta\leq\pi$ that aren't excluded by constraints,
with $\underline{\eta}$ and $\overline{\eta}$ being corresponding
bounds on $\eta\left(\theta\right)$ defined by equation (\ref{eq:DefineEta}).
In equation (\ref{eq:PhiEqOmegaPlusConditional}) the $2\pi$ case
can be omitted if\begin{eqnarray}
\underline{\eta}+\underline{\upsilon} & > & -\pi,\end{eqnarray}
and/or the $-2\pi$ case can be omitted if\begin{eqnarray}
\overline{\eta}+\overline{\upsilon} & < & \pi.\end{eqnarray}
 The same is true for alternative (\ref{eq:PhiIs2PiMinus2PiOr0})
because $\arg\left(1+g(z)/\left(c(z)z^{\alpha}\right)\right)\rightarrow0$.

We can use $\underline{\eta}=\overline{\eta}=\hat{\eta}$ for the
common case where $\eta$ is independent of $\theta$. For unrestricted
complex $z$ we can use $\overline{\upsilon}=\min\left(1,\left|\alpha\right|\right)\pi$
and $\underline{\upsilon}=-\overline{\upsilon}$. For unrestricted
real $z$ we can use the intervals\begin{eqnarray}
\left[\underline{\upsilon},\overline{\upsilon}\right] & = & \begin{cases}
\left[\alpha\pi,0\right] & \mathrm{if}\:-1<\alpha<0,\\
\left[0,0\right] & \mathrm{if}\:\alpha=0,\\
\left[0,\alpha\pi\right] & \mathrm{if}\:0<\alpha\leq1,\\
\left[-\pi,\pi\right] & \mathrm{otherwise}.\end{cases}\end{eqnarray}

Alternative (\ref{eq:PhiEq4Args}) reveals that jumps in $\omega$
near $z=0$ can occur only where $\arg\left(z^{\alpha}\right)=\pi$,
where $\arg\left(c(z)\right)=\pi$, and where $\arg\left(c\left(z\right)z^{\alpha}+g(z)\right)=\pi$.
When $\arg\left(c(z)\right)\neq0$, the three solution sets for these
three equations typically don't pair to cancel or form cusps. Instead
they typically form non-cusp wedges of values $-2\pi$, 0 and $2\pi$.
Therefore, if there are critical angles interior to the included directions,
then:\vspace{-0.1in}

\begin{itemize}
\item Non-zero local $\phi$ occur in infinitely more directions than non-zero
$\omega$. Therefore it is less likely that any constraints will preclude
all of the non-zero wedges or even all of the positive ones or all
of the negative ones.\vspace{-0.1in}

\item It is impossible for both edges of every wedge to coincide and thereby
cancel to give $\phi\equiv0$ for all directions near $z=0$.
\end{itemize}
Thus we often pay dearly for distributing the logarithm over a non-positive
coefficient and its co-factor, which is unnecessary in many applications.
\begin{rem}
For real $z$, $\lim_{z\rightarrow0^{+}}\arg\left(z^{\alpha}\right)=0$
and $\lim_{z\rightarrow0}\arg\left(1+g(z)/\left(c(z)z^{\alpha}\right)\right)\rightarrow0$.
Therefore equation (\ref{eq:PhiIs2PiMinus2PiOr0}) reveals that $\phi=0$
for $z\geq0$ if $\lim_{z\rightarrow0^{+}}\arg\left(c(z)\right)\notin\left\{ -\pi,\pi\right\} $.
When $\lim_{z\rightarrow0^{+}}\arg\left(c(z)\right)\in\left\{ -\pi,\pi\right\} $,
then we can compute $I_{c}\left(g(z)/\left(c(z)z^{\alpha}\right),0\right)$
from definition (\ref{eq:DefineAngleOfs}) to decide whether $\phi$
is $2\pi$, $-2\pi$ or 0 for $z>0$. Similarly $\lim_{z\rightarrow0^{-}}\arg\left(z^{\alpha}\right)=\mathrm{mods}\left(\alpha\pi,2\pi\right)$.
Therefore if\[
\lim_{z\rightarrow0^{-}}\left(\mathrm{mods}\left(\alpha\pi,2\pi\right)+\arg\left(c(z)\right)\right)\notin\left\{ -\pi,\pi\right\} ,\]
then $\phi=0$ for $z\leq0$. Otherwise we can compute $I_{c}\left(g(z)/\left(c(z)z^{\alpha}\right),\pi\right)$
to determine whether $\phi$ is $2\pi$, $-2\pi$ or 0 for $z<0$.
We can then combine these results to obtain $\phi_{\mathrm{real}}$.
\end{rem}

\subsection{\label{sub:Extracting-Dominant-Exponent}Incorrect extraction of
the dominant exponent}

If the dominant exponent is neither 0 nor 1 and either the coefficient
of the dominant term is 1 or we have distributed the logarithm over
it, then we can consider extracting the exponent of the dominant power.

\subsubsection{Principal branch}

The relevant universally correct principal-branch identity for extracting
the exponent of $\ln\left(z^{\alpha}\right)$ is\begin{eqnarray}
\ln\left(u^{\alpha}\right) & \equiv & \alpha\ln(u)+\xi i,\label{eq:LnOfPower}\end{eqnarray}
 where\begin{eqnarray}
\xi & = & \arg\left(u^{\alpha}\right)-\alpha\arg(u),\label{eq:XiEqArgs}\\
 & = & 2\pi\left\lfloor \dfrac{1}{2}-\dfrac{\alpha\arg(u)}{2\pi}\right\rfloor .\label{eq:XiEqFloor}\end{eqnarray}
These alternatives can be derived from identity (\ref{eq:LnAbs})
together with $\ln\left(\left|u^{\alpha}\right|\right)=\alpha\ln\left(\left|u\right|\right)$
and the fact that $\arg\left(u^{\alpha}\right)=\mathrm{mods}\left(\alpha\arg\left(u\right),2\pi\right)$.
Alternative (\ref{eq:XiEqFloor}) is more accurate for approximate
arithmetic and more candid because it is manifestly piecewise-constant
integer multiples of $2\pi$.

Expression $\xi$ can be simplified to a single constant $2n\pi$
if\begin{equation}
(2n-1)\pi\,<\,\alpha\arg(u)\,\leq\,(2n+1)\pi,\end{equation}
for an integer $n$ together with the smallest and largest $\arg(u)$
in $(-\pi,\pi]$ that aren't excluded by any constraint. For example,
$\xi\equiv0$ if \begin{equation}
-\pi\,<\,\alpha\arg(u)\,\leq\,\pi.\label{eq:xiEq0}\end{equation}
If instead $(2n-1)\pi<\alpha\arg(u)\leq(2n+3)\pi$ for all included
$\arg\left(u\right)$, then $\xi$ can be more candidly represented
as\begin{equation}
\begin{cases}
2n\pi & \mathrm{if}\:\alpha\arg(u)\leq(2n+1)\pi,\\
(2n+2)\pi & \mathrm{otherwise}.\end{cases}\end{equation}
Extraction of a fractional power has the benefit of converting a troublesome
fractional power into a benign multiplication by a fraction.

Also, there is more justification for extraction of the dominant exponent
for a hierarchical series wherein the coefficient can be a series
in a logarithm or nested logarithm of $z$.

In our case, $u=z$, and combining this transformation with distribution
over the dominant term and the coefficient thereof gives the overall
transformation\begin{eqnarray}
\ln\left(c\left(z\right)z^{\alpha}+g(z)\right) & \rightarrow & \left(\Psi i+\ln\left(c\left(z\right)\right)+\alpha\ln\left(z\right)\right)+\ln\left(1+\frac{g(z)}{c\left(z\right)z^{\alpha}}\right),\end{eqnarray}
where\begin{eqnarray}
\Psi & \!\!=\!\! & \Phi+\arg\left(z^{\alpha}\right)-\alpha\arg\left(z\right),\label{eq:PsiEqPhiPlus2Args}\\
 & \!\!=\!\! & \Phi+2\pi\!\left\lfloor \dfrac{1}{2}-\dfrac{\alpha\arg(z)}{2\pi}\right\rfloor ,\label{eq:PsiEqPhiPlusFloor}\\
 & \!\!=\!\! & \Omega+\arg\left(c(z)z^{\alpha}\right)-\arg\left(c(z)\right)-\alpha\arg\left(z\right),\label{eq:PsiEqOmegaPlusArgs}\\
 & \!\!=\!\! & \Omega+2\pi\!\left\lfloor \dfrac{1}{2}-\dfrac{\alpha\arg(z)}{2\pi}\right\rfloor +\begin{cases}
2\pi & \mathrm{if}\:\arg\left(c(z)\right)+\arg\left(z^{\alpha}\right)\leq-\pi,\\
-2\pi & \mathrm{if}\:\arg\left(c(z)\right)+\arg\left(z^{\alpha}\right)>\pi,\\
0 & \mathrm{otherwise},\end{cases}\label{eq:PsiEqOmegaPlusFloorPlusConditional}\\
 & \!\!=\!\! & \arg\left(c\left(z\right)z^{\alpha}+g(z)\right)-\arg\left(1+\frac{g(z)}{c\left(z\right)z^{\alpha}}\right)-\arg\left(c\left(z\right)\right)-\alpha\arg\left(z\right),\label{eq:PsiEq4Args}\\
 & \!\!=\!\! & 2\pi\!\left\lfloor \dfrac{1}{2}\!-\!\dfrac{\alpha\arg(z)}{2\pi}\right\rfloor +\begin{cases}
2\pi & \!\mathrm{\!\! if}\:\arg\!\left(c(z)\right)\!+\!\arg\!\left(z^{\alpha}\right)\!+\!\arg\!\left(1\!+\!\dfrac{g(z)}{c(z)z^{\alpha}}\right)\leq-\pi,\\
\!-2\pi & \mathrm{\!\!\! if}\:\arg\!\left(c(z)\right)\!+\!\arg\!\left(z^{\alpha}\right)\!+\!\arg\!\left(1\!+\!\dfrac{g(z)}{c(z)z^{\alpha}}\right)>\pi,\\
0 & \mathrm{\!\!\! otherwise}.\end{cases}\label{eq:PsiEqFloorPlusConditional}\end{eqnarray}
Let $\psi$ denote a perhaps simpler local version of $\Psi$ near
$z=0$. Alternative (\ref{eq:PsiEqPhiPlusFloor}) is advantageous
when $\psi$ is non-constant but $\phi$ is simple. Alternative (\ref{eq:PsiEqOmegaPlusFloorPlusConditional})
is advantageous when $\psi$ and $\phi$ are non-constant but $\omega$
is simple. Otherwise alternative (\ref{eq:PsiEqFloorPlusConditional})
is most candid and least subject to catastrophic cancellation.

Equation (\ref{eq:PsiEq4Args}) reveals that the critical angles are
the union of the solutions to $\arg\left(z\right)=\pi$, $\arg\left(c(z)z^{\alpha}\right)=\pi$,
and (if $c\left(z\right)$ depends on $z$) $\arg\left(c(z)\right)=\pi$.
The piecewise-constant values of $\Psi$ are also multiples of $2\pi$,
but no longer limited to be one of $-2\pi$, 0 or $2\pi$. Therefore
with or without constraints on $\arg\left(z\right)$, the local version
$\psi$ is even less likely than $\phi$ to be identically 0 or otherwise
constant, and likewise for $\psi_{\mathrm{real}}$. However, we will
see that the use of $\Psi$ rather than $\Omega$ is mandatory for
the fractional power of a series. Therefore it is important to simplify
$\Psi$ as much as is practical: The floor term in alternatives (\ref{eq:PsiEqPhiPlusFloor}),
(\ref{eq:PsiEqOmegaPlusFloorPlusConditional}) and (\ref{eq:PsiEqFloorPlusConditional})
can be simplified as discussed at the beginning of this sub-subsection.
Subsection \ref{sub:ExtractionOfDomCoef} describes how to simplify
the conditional term in alternative (\ref{eq:PsiEqOmegaPlusFloorPlusConditional}).
Simplification of the conditional term in alternative (\ref{eq:PsiEqFloorPlusConditional})
is similar because $\arg\left(1+g(z)/\left(c(z)z^{\alpha}\right)\right)\rightarrow0$.

For real $z$, we can compute a local version of $\Psi$, $\psi_{\mathrm{real}}$,
by computing $\phi_{\mathrm{real}}$, then adding\[
2\pi\left\lfloor \dfrac{1}{2}-\dfrac{\alpha}{2}\right\rfloor \]
to the case for negative $z$.

\subsubsection{\label{sub:RealBranchLnFracPow}The real branch for ln of fractional
powers}

The \textsl{Derive} computer algebra system offers a \textsf{Branch}
control variable that, if assigned the value \textsf{Real} causes
fractional powers having odd denominators to use the real rather than
principal branch when a radicand is negative. TI-Nspire bundles this
choice into its \textsf{Real} mode.

For example, in \textsf{Real} mode $\left(-1\right)^{1/3}\rightarrow-1$
rather than $1/2+i\sqrt{3}/2$, and $\left(-1\right)^{2/3}\rightarrow1$
rather than $-1/2+i\sqrt{3}/2$. This default option is much appreciated
by students and faculty who fear or loathe non-real numbers, which
is a majority of TI's customers most of the time. However, this choice
is incompatible with formulas (\ref{eq:LnOfPower}) through (\ref{eq:XiEqArgs}).
For example, these principal-branch formulas give\begin{eqnarray*}
\ln\left(x^{2/3}\right) & \rightarrow & \frac{2}{3}\ln\left(x\right)+\left(\arg\left(x^{2/3}\right)-\frac{2}{3}\arg\left(x\right)\right)i.\end{eqnarray*}

If we subsequently substitute -8 for $x$, then this result simplifies
to\begin{eqnarray*}
\frac{2}{3}\ln\left(-8\right)+\left(\left(\arg\left(\left(-8\right)^{2/3}\right)-\frac{2}{3}\arg\left(-8\right)\right)i\right) & \rightarrow & 2\ln\left(2\right)+\frac{2}{3}\pi i,\end{eqnarray*}
whereas for the real branch, $\ln\left((-8)^{2/3}\right)\rightarrow\ln\left(4\right)\rightarrow2\ln\left(2\right).$ 

The easiest way to simplify logarithms of fractional powers using
the real branch is to refrain from extracting reduced fractional exponents
having an odd denominator unless the sign of the base can be determined
-- perhaps with the assistance of a constraint on $\arg\left(z\right)$.
Moreover, for real-branch mode we should also refrain from extracting
integer exponents that are even. For example, we shouldn't extract
the exponent 12 from $\ln\left(x^{12}\right)$ for a real $x$ because\begin{eqnarray*}
\ln\left(x^{12}\right) & \rightarrow & 12\ln\left(x\right)+\left(\arg\left(x^{12}\right)-12\arg\left(x\right)\right)i\\
 & \rightarrow & 12\ln\left(x\right)-\begin{cases}
12\pi i & \mathrm{if}\: x<0,\\
0 & \mathrm{otherwise.}\end{cases}\end{eqnarray*}

Even though this result is real for all real $x$, this appearance
of $i$ in the result is unwelcome to most real-mode customers. Even
principal-branch complex-mode customers would rather not see $i$
in an expression if it can as concisely or more concisely be expressed
without $i$ -- especially if the result is real for all real values
of interest for any variables therein. However, to help reduce the
number of distinct logands in a result for $x\in\mathbb{R}$, we can
extract all of an even exponent but 2:\begin{eqnarray*}
\ln\left(x^{12}\right) & \rightarrow & 6\ln\left(x^{2}\right).\end{eqnarray*}

Another alternative for real $x$ is $\ln\left(x^{12}\right)\rightarrow12\ln\left(|x|\right)$.
However, absolute-value functions are troublesome and best avoided
where possible. For example, if in the same series we use\begin{eqnarray*}
\ln\left(x^{3}\right) & \rightarrow & 3\ln\left(x\right)-\begin{cases}
2\pi i & \mathrm{if}\: x<0,\\
0 & \mathrm{otherwise,}\end{cases}\end{eqnarray*}
then we obtain a series that contains both $\ln\left(x\right)$ and
$\ln\left(|x|\right)$.

\section{Branch bugs for fractional powers\label{sec:BugsForFracPows}}

Some series implementations can give incorrect series results for
fractional powers. For example, with $\left(z^{2}+z^{3}\right)^{3/2}$
most implementations currently incorrectly give the equivalent of\[
\sum_{k=0}^{\infty}\dfrac{3\left(-1\right)^{k}\left(2k\right)!}{(2k-3)(2k-1)k!^{2}4^{k}}z^{k+3}\]
or a truncated version of it. This series is incorrect by a factor
of -1 left of $x=\left(\sqrt{3y^{2}+1}-1\right)/3$ for $z=x+iy$.

Table \ref{tab:SeriesForFracPow} lists one or more correct alternative
multiplicative correction factors for the Puiseux series of some fractional-power
expressions expanded about complex $z=0$ or real $x=0$.

\begin{table}[h]
\caption{For $n\!\geq\!4$, $x,y\!\in\mathbb{\! R}$, $z\!=\! x\!+\! iy$:
series$\left(u^{\beta},\,\mbox{var}\!=\!0,\, o\!\left(\mathrm{var}^{n}\right)\right)\rightarrow(-1)^{\lambda}\left(S\!+\!\cdots\right)$\label{tab:SeriesForFracPow}}
\begin{tabular}{|c|c|c|c|c|}
\hline 
\# & $u^{\beta}$ & $\lambda$, or $L=(-1)^{\lambda}$ & $S$ & why\tabularnewline
\hline
\hline 
\negthinspace{}\negthinspace{}1a\negthinspace{}\negthinspace{} & $\left(z^{2}+z^{3}\right)^{\nicefrac{3}{2}}$ & $\!\! L\!=\!\begin{cases}
\!1, & \left(\Im\!\left(z\right)\!\geq\!0\wedge\Re\!\left(z\right)\!\geq\!\Im\!\left(z\right)^{2}\!/2\!+\!...\!+\! o\!\left(\Im\!\left(z\right)^{n}\right)\right)\\
\! & \!\!\!\!\vee\left(\Im\!\left(z\right)\!<\!0\wedge\Re\!\left(z\right)\!>\!\Im\!\left(z\right)^{2}\!/2\!+\!...\!+\! o\!\left(\Im\!\left(z\right)^{n}\right)\right)\\
\!-1, & \mathrm{\!\!\! otherwise}\end{cases}\!\!$ & $z^{3}$ & $\begin{array}{c}
(\ref{eq:L1HatForHalfInteger})\\
(\ref{eq:DefineL2Hat})\\
\mathrm{\! rem}.\:\ref{rem:ExplictCurveAndSeriesThereof}\!\end{array}$\tabularnewline
\hline 
\negthinspace{}\negthinspace{}1b\negthinspace{}\negthinspace{} &  & $L=\begin{cases}
1 & \mathrm{if\:}\left(\Re\!\left(z\right)>0\,\vee\,\Re\!\left(z\right)=0\,\wedge\,\Im\!\left(z\right)\geq0\right)\\
 & \quad=\left(-\pi<\arg\left(z^{2}\right)+\arg\left(1+z\right)\leq\pi\right)\\
-1 & \mathrm{otherwise}\end{cases}$ & $z^{3}$ & $\begin{array}{c}
(\ref{eq:L1HatForHalfInteger})\\
(\ref{eq:DefineL2Hat})\end{array}$\tabularnewline
\hline 
\negthinspace{}\negthinspace{}1c\negthinspace{}\negthinspace{} &  & $\lambda=\frac{3}{2\pi}\left(\arg\left(z^{2}+z^{3}\right)-\arg\left(1+z\right)-2\arg\left(z\right)\right)$ & $z^{3}$ & \negthinspace{}\negthinspace{}(\ref{eq:PsiEq4Args},\ref{eq:LambdaViaPsi})\negthinspace{}\negthinspace{}\tabularnewline
\hline 
\negthinspace{}\negthinspace{}2a\negthinspace{}\negthinspace{} & $\left(z^{2}-iz^{3}\right)^{\nicefrac{3}{2}}$ & $L=\begin{cases}
1 & \mathrm{if}\:\Re\left(z\right)>0\,\vee\,\Re\left(z\right)=0\,\wedge\,\Im\left(z\right)\geq0\\
-1 & \mathrm{otherwise}\end{cases}$ & $z^{3}$ & $\begin{array}{c}
(\ref{eq:PsiEqPhiPlusFloor})\\
(\ref{eq:L1HatForHalfInteger})\end{array}$\tabularnewline
\hline 
\negthinspace{}\negthinspace{}2b\negthinspace{}\negthinspace{} &  & $\lambda=\left\lfloor \nicefrac{1}{2}-\left(\arg z\right)/2\right\rfloor $ & $z^{3}$ & (\ref{eq:PsiEqPhiPlusFloor})\tabularnewline
\hline 
\negthinspace{}\negthinspace{}3a\negthinspace{}\negthinspace{} & $\left(z^{2}+z^{3}\right)^{\nicefrac{7}{4}}$ & $\!\! L\!=\!\begin{cases}
\!1, & \left(\Im\!\left(z\right)\!\geq\!0\wedge\Re\!\left(z\right)\!\geq\!\Im\!\left(z\right)^{2}\!/2\!+\!...\!+\! o\!\left(\Im\!\left(z\right)^{n}\right)\right)\\
\! & \!\!\!\!\vee\left(\Im\!\left(z\right)\!<\!0\wedge\Re\!\left(z\right)\!>\!\Im\!\left(z\right)^{2}\!/2\!+\!...\!+\! o\!\left(\Im\!\left(z\right)^{n}\right)\right)\\
\! i, & \left(\Im\!\left(z\right)\!\geq\!0\wedge\Re\!\left(z\right)\!<\!\Im\!\left(z\right)^{2}\!/2\!+\!...\!+\! o\!\left(\Im\!\left(z\right)^{n}\right)\right)\\
\!-i, & \mathrm{\!\!\! otherwise}\end{cases}\!\!$ & $z^{7/2}$ & $\begin{array}{c}
(\ref{eq:DefineL1})\\
(\ref{eq:DefineL2})\\
\!\mathrm{rem}.\:\ref{rem:ExplictCurveAndSeriesThereof}\!\end{array}$\tabularnewline
\hline 
\negthinspace{}\negthinspace{}3b\negthinspace{}\negthinspace{} &  & $\lambda=\frac{7}{4\pi}\left(\arg\left(z^{2}+z^{3}\right)-\arg\left(1+z\right)-2\arg\left(z\right)\right)$ & $z^{7/2}$ & \negthinspace{}\negthinspace{}(\ref{eq:PsiEq4Args},\ref{eq:LambdaViaPsi})\negthinspace{}\negthinspace{}\tabularnewline
\hline 
4 & $\!\left(-1\!+\! iz^{\nicefrac{1}{4}}\!+\! z\right)^{\nicefrac{3}{2}}\!$ & $L=1$ & $-i$ & \negthinspace{}\negthinspace{}(\ref{eq:tauForomegaEq0},\ref{eq:PsiEqOmegaPlusFloorPlusConditional})\negthinspace{}\negthinspace{}\tabularnewline
\hline 
5 & $\!\left(-1\!+\! iz^{\nicefrac{1}{2}}\!+\! z\right)^{\nicefrac{3}{2}}\!$ & $L=1$ & $-i$ & rem.\negthinspace{} \ref{rem:BorderlineCase}\tabularnewline
\hline 
\negthinspace{}\negthinspace{}6a\negthinspace{}\negthinspace{} & $\!\left(-1\!-\! z^{2\!}-\! z^{3}\right)^{\nicefrac{3}{2}}\!$ & $L=\begin{cases}
-1 & \mathrm{if\:}\left(y\geq0\wedge x<y^{2}/2-3y^{4}/8+...\right)\vee\\
 & \quad\left(y<0\wedge x\geq y^{2}/2-3y^{4}/8+...\right)\\
1 & \mathrm{otherwise}\end{cases}$ & $-i$ & $\begin{array}{c}
(\ref{eq:L1HatForHalfInteger})\\
(\ref{eq:DefineL2Hat})\\
\mathrm{\! rem}.\:\ref{rem:ExplictCurveAndSeriesThereof}\!\end{array}$\tabularnewline
\hline 
\negthinspace{}\negthinspace{}6b\negthinspace{}\negthinspace{} &  & $\!\lambda=\frac{3}{2\pi}\left(\arg\left(-1\!-\! z^{2}\!-\! z^{3}\right)-\arg\left(1\!+\! z^{2}\!+\! z^{3}\right)\!-\!\pi\right)\!$ & $-i$ & \negthinspace{}\negthinspace{}(\ref{eq:OmegaPiOfImGOnC},\ref{eq:LambdaViaPsi})\negthinspace{}\negthinspace{}\tabularnewline
\hline 
7 & $\!\left(c\!+\! z^{2}\right)^{\nicefrac{3}{4}}|\: c\!\neq\!0\!$ & \negthinspace{}$L=\begin{cases}
i & \mathrm{if}\:\arg c=\pi\wedge\Im\left(z^{2}/c\right)>0\\
1 & \mathrm{otherwise}\end{cases}$ & $c^{\nicefrac{3}{4}}$ & $\begin{array}{c}
\!\mathrm{rem}.\:\ref{rem:CasesForLiteralDominantCoefs}\!\\
(\ref{eq:OmegaPiOfImGOnC})\end{array}$\tabularnewline
\hline 
8 & $\!\left(-z^{-\nicefrac{1}{2}}\right)^{\nicefrac{1}{2}}\!$ & $L=\begin{cases}
-1 & \mathrm{if}\:\arg\left(z\right)<0\\
1 & \mathrm{otherwise}\end{cases}$ & $iz^{-\nicefrac{1}{4}}$ & $\begin{array}{c}
\mathrm{\! rem}.\:\ref{rem:OmegaEq0IfGEq0}\!\\
(\ref{eq:PsiEqOmegaPlusFloorPlusConditional})\end{array}$\tabularnewline
\hline 
\negthinspace{}\negthinspace{}9\negthinspace{}\negthinspace{} & $\left(\ln\left(z\right)+z\right)^{\nicefrac{3}{2}}$ & $L=1$ & $\ln\left(z\right)^{\nicefrac{3}{2}}$ & (\ref{eq:OmegaPiOfImGOnC})\tabularnewline
\hline 
\negthinspace{}\negthinspace{}10\negthinspace{}\negthinspace{} & $\left(iz+z^{2}\right)^{\nicefrac{3}{2}}$ & $L=\begin{cases}
-1 & \mbox{if}\;\Re\left(z\right)<0\,\wedge\,\Im\left(z\right)\geq0\\
1 & \mbox{otherwise}\end{cases}$ & $\!\left(-1\right)^{\frac{3}{4}}z^{\frac{3}{2}}\!$ & $\begin{array}{c}
(\ref{eq:DefineAngleOfs})\\
(\ref{eq:DefineL2HatUsingOmega})\end{array}$\tabularnewline
\hline 
\negthinspace{}\negthinspace{}11\negthinspace{}\negthinspace{} & $\left(z^{-1}+1\right)^{\nicefrac{3}{2}}$ & $L=\begin{cases}
-1 & \mbox{if}\; z<0\\
1 & \mbox{otherwise}\end{cases}$ & $z^{\nicefrac{-3}{2}}$ & $\begin{array}{c}
\mathrm{\! prop}.\:\ref{pro:omegaForExponentDivisibility}\!\\
(\ref{eq:DefineL2HatUsingOmega})\end{array}$\tabularnewline
\hline 
\negthinspace{}\negthinspace{}12\negthinspace{}\negthinspace{} & $\left(-2+x\right)^{\nicefrac{3}{2}}$ & $L=1$ & $-i$ & \negthinspace{}prop.\negthinspace{} \ref{pro:omegaRealForRealCfAndIntegerExpon}\negthinspace{}\tabularnewline
\hline 
\negthinspace{}\negthinspace{}13\negthinspace{}\negthinspace{} & $\left(-1-\sqrt{x}\right)^{\nicefrac{7}{2}}$ & $L=\begin{cases}
1 & \mathrm{if}\: x\geq0\\
-1 & \mathrm{otherwise}\end{cases}$ & $-i$ & (\ref{eq:DefineAngleOfs})\tabularnewline
\hline 
\negthinspace{}\negthinspace{}14a\negthinspace{}\negthinspace{} & $\left(x^{\nicefrac{4}{3}}+x^{2}\right)^{\nicefrac{3}{2}}$ & $L=\begin{cases}
-1 & \mbox{if}\; x<0\\
1 & \mbox{otherwise}\end{cases}$ & $x^{2}$ & $\!\left(\ref{eq:DefineAngleOfs},\!\ref{eq:DefineL2HatUsingOmega}\right)\!$\tabularnewline
\hline 
\negthinspace{}\negthinspace{}14b\negthinspace{}\negthinspace{} & real branch & $L=1$ & $x^{2}$ & \negthinspace{}rem.\negthinspace{} \ref{rem:omegaRealForOddDenominators}\negthinspace{}\tabularnewline
\hline
\end{tabular}
\end{table}
The probable causes of an incorrect result are related to those for
logarithms of series. The usual algorithm for computing a numeric
power $\beta$ of a series requires that the dominant term be 1. Therefore
if the series of the radicand is $c\left(z\right)z^{\alpha}+g(z)$,
with $c\left(z\right)z^{\alpha}$ being the dominant term, we must:\vspace{-0.1in}
 
\begin{enumerate}
\item Factor out $c\left(z\right)z^{\alpha}$.\vspace{-0.1in}

\item Distribute the exponent $\beta$ over the three factors $\left(z\right),$
$cz^{\alpha}$ and $1+g\left(z\right)/\left(c\left(z\right)z^{\alpha}\right)$.\vspace{-0.1in}

\item Transform $\left(z^{\alpha}\right)^{\beta}$ to $z^{\alpha\beta}$.\vspace{-0.1in}

\item Compute the series for $\left(1+g\left(z\right)/\left(c\left(z\right)z^{\alpha}\right)\right)^{\beta}$
by the usual algorithm.\vspace{-0.1in}

\item Distribute $c\left(z\right)^{\beta}z^{\alpha\beta}$ over the terms
computed in step 5.
\end{enumerate}
Steps 2 and 3 can contribute to an angular rotation factor of the
form $(-1)^{\cdots}$ that should also be distributed in step 5.

\subsection{Principal-branch series of fractional powers \label{sub:Principal-branch-series-of-frac-pow}}

Zippel's formula \cite{ZippelMacsymaConference} generalizes to the
following universal principal-branch formula for the distribution
of real exponents over products:\begin{eqnarray}
(uv)^{\beta} & \equiv & \left(-1\right)^{\delta}u^{\beta}v^{\beta}\end{eqnarray}
where

\begin{eqnarray}
\delta & = & \frac{\beta}{\pi}\left(\arg\left(\left(uv\right)^{\beta}\right)-\arg(u^{\beta})-\arg(v^{\beta})\right).\end{eqnarray}

For real $\alpha$ and $\beta$, a universal principal-branch formula
for transforming a power of a power to an unnested power is\begin{eqnarray}
\left(w^{\alpha}\right)^{\beta} & \rightarrow & \left(-1\right)^{\zeta}w^{\alpha\beta}\end{eqnarray}
where\begin{eqnarray}
\zeta & := & \dfrac{\beta}{\pi}\left(\arg\left(w^{\alpha}\right)-\alpha\arg\left(w\right)\right).\end{eqnarray}
This can be derived from the identities

\begin{eqnarray}
\left|p\right| & \equiv & (-1)^{-\arg\left(p\right)/\pi}p,\label{eq:absViaArg}\\
\left|q^{\alpha}\right|^{\beta} & \equiv & \left|q\right|^{\alpha\beta}.\end{eqnarray}

Applying these formulas to our case gives\begin{eqnarray}
\left(c\left(z\right)z^{\alpha}+g\left(z\right)\right)^{\beta} & \rightarrow & \left(c(z)z^{\alpha}\left(1+\frac{g\left(z\right)}{c\left(z\right)z^{\alpha}}\right)\right)^{\beta}\nonumber \\
 & \rightarrow & (-1)^{\Lambda}c\left(z\right)^{\beta}z^{\alpha\beta}\left(1+\frac{g\left(z\right)}{c\left(z\right)z^{\alpha}}\right)^{\beta}\end{eqnarray}
where modulo 2, which suffices for $(-1)^{\Lambda}$, \begin{eqnarray}
\Lambda & \equiv & \dfrac{\beta}{\pi}\left(\arg\left(c\left(z\right)z^{\alpha}+g\left(z\right)\right)-\arg\left(1\!+\!\dfrac{g(z)}{c(z)z^{\alpha}}\right)-\arg(c\left(z\right))-\alpha\arg(z)\right)\nonumber \\
 & = & \dfrac{\beta}{\pi}\Psi,\label{eq:LambdaViaPsi}\end{eqnarray}
with $\Psi$ defined by alternatives (\ref{eq:PsiEqPhiPlus2Args})
through (\ref{eq:PsiEqFloorPlusConditional}). Let\begin{eqnarray}
L & := & (-1)^{\Lambda}.\label{eq:DefineL}\\
 & = & L_{1}L_{2},\label{eq:FracPowFloorFactor}\end{eqnarray}
where\begin{eqnarray}
L_{1} & = & \left(-1\right)^{2\beta\left\lfloor \dfrac{1}{2}\,-\,\dfrac{\alpha\arg\left(z\right)}{2\pi}\right\rfloor },\label{eq:DefineL1}\\
L_{2} & = & (-1)^{\beta\Omega/\left(2\pi\right)}\begin{cases}
(-1)^{2\beta} & \mathrm{if}\:\arg(c\left(z\right))+\arg(z^{\alpha})\leq-\pi,\\
(-1)^{-2\beta} & \mathrm{if}\:\arg(c\left(z\right))+\arg(z^{\alpha})>\pi,\\
1 & \mathrm{otherwise},\end{cases}\label{eq:DefineL2UsingOmega}\\
 & = & \begin{cases}
(-1)^{2\beta} & \mathrm{if}\:\arg(c\left(z\right))+\arg(z^{\alpha})+\arg\left(1\!+\!\dfrac{g(z)}{c(z)z^{\alpha}}\right)\leq-\pi,\\
(-1)^{-2\beta} & \mathrm{if}\:\arg(c\left(z\right))+\arg(z^{\alpha})+\arg\left(1\!+\!\dfrac{g(z)}{c(z)z^{\alpha}}\right)>\pi,\\
1 & \mathrm{otherwise}.\end{cases}\label{eq:DefineL2}\end{eqnarray}

Alternative (\ref{eq:DefineL2UsingOmega}) is preferable for $L_{2}$
when $\omega$ is constant. Techniques described in subsection \ref{sub:Extracting-Dominant-Exponent}
can further simplify the floor sub-expression in $L_{1}$. Also, the
exponents of -1 in these formulas can be simplified and canonicalized
by replacing them with their near-symmetric residue modulo 2.
\begin{rem}
\label{rem:IntegerPowers}In the common case where $\beta$ is integer,
then\[
(-1)^{2\beta}\,=\,(-1)^{-2\beta}\,=\,(-1)^{2\beta\mathrm{\left\lfloor 1/2-\alpha\arg\left(z\right)/(2\pi)\right\rfloor }}\,=\,1,\]
so $L$ can be simplified to 1.
\end{rem}
\,

\vspace{-0.016in}

\begin{rem}
\label{rem:HalfIntegerPowers}In the next most common case where $\beta$
is a half-integer, $L_{1}$ can be expressed more candidly as\begin{eqnarray}
\hat{L}_{1} & = & \begin{cases}
1 & \mathrm{if}\:\mathrm{mods}\left(\pi-\alpha\arg(z),\,2\pi\right)\geq0,\\
-1 & \mathrm{otherwise},\end{cases}\label{eq:L1HatForHalfInteger}\end{eqnarray}
 and $L_{2}$ can be simplified to\begin{eqnarray}
\hat{L}_{2} & = & (-1)^{\beta\Omega/\left(2\pi\right)}\begin{cases}
1 & \mathrm{if}\:-\pi<\arg(c\left(z\right))+\arg(z^{\alpha})\leq\pi\\
-1 & \mathrm{otherwise}.\end{cases}\label{eq:DefineL2HatUsingOmega}\\
 & = & \begin{cases}
1 & \mathrm{if}\:-\pi<\arg(c\left(z\right))+\arg(z^{\alpha})+\arg\left(1\!+\!\dfrac{g(z)}{c(z)z^{\alpha}}\right)\leq\pi,\\
-1 & \mathrm{otherwise}.\end{cases}\label{eq:DefineL2Hat}\end{eqnarray}
It is easy and worthwhile to test for these cases. Alternative (\ref{eq:DefineL2HatUsingOmega})
is preferable when $\omega$ is a constant.
\end{rem}
Techniques described in subsection \ref{sub:Extracting-Dominant-Exponent}
can further simplify local equivalents to $L_{2}$ and $\hat{L}_{2}$
-- particularly for real $z$.

Also, if the local equivalent simplifies to a piecewise constant that
is one value $\ell_{0}$ at $z=0$ for which $0\,\ell_{0}$ simplifies
to 0, but another value $\ell_{*}$ everywhere else, then we can simply
use $\ell_{*}$ for all of the terms having positive degree, because
those powers of $z$ are 0 at $z=0$. This can dramatically simplify
the result because often there are no non-positive degree terms or
only one. For example, with real $x$ there is only one non-positive
degree term for\begin{eqnarray*}
\mathrm{series\!}\left(\!\left(-1+x-\! ix^{2}\right)^{7/2}\!,\, x=0,\, o\!\left(x^{2}\right)\right) & \!\!\rightarrow\!\! & \left(\!\begin{cases}
-i & \mathrm{\!\!\! if}\: x=0\\
i & \mathrm{\!\!\! otherwise}\end{cases}\right)-\dfrac{7}{2}ix+\dfrac{35}{8}ix^{2}+o\!\left(x^{2}\right).\end{eqnarray*}

Otherwise if $L$ doesn't simplify to a constant, then we have to
include a perhaps complicated \textsl{\emph{factor}} multiplying \textsl{every
term} of the result, because extracting the dominant coefficient and
exponent are mandatory for numeric powers of series. This is significantly
more annoying than the possible one piecewise-constant term for a
logarithm of a series. If we choose to distribute non-constant $L$
over every term, then the result is significantly bulkier and less
intelligible. If instead we choose to return an undistributed product
of this factor with a sum of terms, then we have denied the user what
they implicitly requested by using a function named series -- a \textsl{sum}
of terms.

If we are implementing the full generality of a hierarchical series,
then such recursively-represented series have some appeal, because
the jumps in the rotation angle make it have some properties of some
essential singularities, which dominate any power of $z$. However,
one can argue that such jumps belong in the coefficients because the
magnitude of the rotation factor is always 1, whereas the logarithmic
singularities that we already allow in the coefficients have infinite
magnitude at $z=0$.

It is easy for a human user to use an $\textrm{expand}\left(\ldots\right)$
function to distribute the factor over the terms after seeing an unexpanded
result. However, the user might be other functions whose authors must
be knowledgeable enough to realize that they should always apply $\textrm{expand}\left(\ldots,z\right)$
to the result of $\textrm{series}\left(\ldots\right)$ because it
could be an expandable product.

If we are not implementing the full generality of hierarchical series,
then a more serious disadvantage of not automatically distributing
the factor is that it requires a special field for a multiplicative
factor in the series data structure; and one such specialized exceptional
field is inadequate for adding series containing such multiplicative
factors if the series have different dominant exponents.

Expressions such as $(-1)^{2\beta}$ can alternatively be expressed
as $L:=e^{2\beta\pi i}$ or represented that way internally. However
for presentation in results, $(-1)^{2\beta}$ avoids a perhaps-unnecessary
$i$ and more obviously indicates that the factor has magnitude 1.
\begin{rem}
\label{rem:CombiningConditionalsForRealz}Whenever $z$ is real and
$L$ contains more than one piecewise factor, their product can often
be combined into a single factor by separately simplifying their product
under the alternative constraints $z<0$, $z=0$, and $z>0$, then
forming a single piecewise or constant factor accordingly if the three
values are constants. This process often simplifies the product to
the constant 1 when $\beta$ is a half integer.
\end{rem}

\subsection{Real-branch series of fractional powers}

For computing a fractional power of a series we must fully distribute
the exponent over the dominant coefficient and dominant power, then
combine the two exponents of the latter into a product. A way to accomplish
this correctly for real-branch mode is to use in this mode the following
rewrite rules for reduced exponents for a real expression $t$ together
with integers $m$ and $n$: \begin{eqnarray}
\arg\left(t^{2m/(2n+1)}\right) & \rightarrow & 0,\\
\arg\left(t^{\left(2m+1\right)/(2n+1)}\right) & \rightarrow & \arg\left(t\right).\end{eqnarray}

\section{Branch bugs for inverse trig \& hyperbolic series\label{sec:Branch-bugs-for-arc}}

Kahan \cite{Kahan} is often cited as a standard for defining the
branch cuts of the inverse trigonometric and inverse hyperbolic functions,
together with the principal values on those branch cuts and on the
branch points that end them. In comparison to \textsl{Derive}, his
definition of $\arctan\left(\ldots\right)$ has the disadvantage of
violating the useful identity\begin{eqnarray}
\arctan\left(\dfrac{z}{\sqrt{1-z^{2}}}\right) & \equiv & \arcsin\left(z\right).\end{eqnarray}
However, this section uses Kahan's definitions, which are used by
TI-Nspire.

This section doesn't discuss the six secondary functions such as $\mathrm{arcsec}\left(z\right):=\arccos\left(1/z\right)$
and $\mathrm{arccot}\left(z\right):=\pi/2-\arctan\left(z\right)$
because their series are easily computed from such definitions. Beware,
however, that different mathematical software might use different
definitions. For example, \textsl{Mathematica} uses\[
\mathrm{arccot}\left(z\right):=\begin{cases}
\pi/2 & \mathrm{if}\: z=0,\\
\dfrac{i}{2}\left(\ln\left(1-i/z\right)-\ln\left(1+i/z\right)\right) & \mathrm{otherwise},\end{cases}\]
which is not equivalent everywhere to $\pi/2-\arctan\left(z\right)$.

One way to compute series for the primary inverse trigonometric and
inverse hyperbolic functions is by their integral definitions. However,
the handling of non-constant coefficients that can include nested
logarithms and piecewise constants is then problematic, and this method
alone doesn't specify the degree-0 term of the result.

Consequently, this article instead uses Kahan's definitions in terms
of logarithms and powers. 

Many of the tests in this section entail comparing the imaginary part
of a series $V$ with 0. For both real and complex $z$ these tests
can often be simplified by determining critical angles and perhaps
also computing $I_{c}\left(V,\theta_{c}\right)$, as described in
Sub-subsections \ref{sub:lnNon0DominantExponent} and \ref{sub:The-degree-0-case}.
Some of the following formulas instead entail comparing the real part
of a series $V$ with 0. They often can similarly be simplified by
determining critical angles and perhaps also computing $I_{c}\left(iV,\theta_{c}\right)$,
which maps the real part to an imaginary part.

\subsection{Series for arctanh}

The inverse hyperbolic tangent of a series $U$ can be computed from
the identity\begin{eqnarray}
\mbox{arctanh}\left(U\right) & \equiv & \frac{\ln\left(1+U\right)-\ln\left(1-U\right)}{2}.\label{eq:atanhViaLn}\end{eqnarray}

The logarithms in formula (\ref{eq:atanhViaLn}) can contribute expressions
involving $\omega$, $\phi$ or $\psi$ as discussed in Section \ref{sec:Branch-bugs-for-ln}.
The result is candid in most cases. However, if series $U=cz^{\alpha}+g(z)$
has a negative dominant exponent $\alpha$, then the zero-degree term
computed by (\ref{eq:atanhViaLn}) is\begin{equation}
\ln\left(cz^{\alpha}\right)-\ln\left(-cz^{\alpha}\right)+\left(\Omega\left(U+1\right)+\Omega\left(-U+1\right)\right).\end{equation}
This expression can and should be replaced with the simpler local
equivalent

\begin{equation}
\dfrac{\pi i}{2}\begin{cases}
\pm1 & \mathrm{if\:}z=0,\\
1 & \mathrm{if\:}\Im\left(U\right)>0\:\vee\:\Im\left(U\right)=0\,\wedge\,\Re\left(U\right)<0\\
-1 & \mathrm{otherwise},\end{cases},\end{equation}
which is also less prone to rounding errors.

\subsection{Series for arctan\label{sub:Series-for-arctan}}

The inverse tangent of a series $U$ can be computed from the identity\begin{eqnarray}
\arctan\left(U\right) & \equiv & -i\,\mbox{arctanh}\left(iU\right).\label{eq:atanViaAtanh}\end{eqnarray}

The degree-0 term of $\arctan\left(U\right)$ is\begin{equation}
\begin{cases}
0 & \!\!\!\mathrm{if\:}\alpha>0,\\
\pm\pi/2 & \mathrm{\!\!\! if\:}\alpha<0\,\wedge\, z=0,\\
\pi/2 & \mathrm{\!\!\! if\:}\alpha<0\,\wedge\,\left(\Re\left(u\right)>0\,\vee\,\Re\left(u\right)=0\,\wedge\,\Im\left(u\right)>0\right),\\
-\pi/2 & \mathrm{\!\!\! if\:}\alpha<0\,\wedge\,\left(\Re\left(u\right)<0\,\vee\,\Re\left(u\right)=0\,\wedge\,\Im\left(u\right)<0\right),\\
\frac{1}{2}\left(-i\ln\left(\frac{ibz^{\sigma}}{2}\right)+\Omega\left(ibz^{\sigma}+ih(z)\right)\right) & \mathrm{\!\!\! if\:}U=i+bz^{\sigma}+h(z),\\
\frac{1}{2}\left(i\ln\left(\frac{-ibz^{\sigma}}{2}\right)-\Omega\left(-ibz^{\sigma}-ih(z)\right)\right) & \mathrm{\!\!\! if\:}U=-i+bz^{\sigma}+h(z),\\
i\,\mathrm{arctanh}(t)+\frac{1}{2}\Omega_{\pi,0}\left(1-t+ig(z)\right) & \mathrm{\!\!\! if}\: U=ti+g(z)\,\wedge\, t>1,\\
i\,\mathrm{arctanh}(t)-\frac{1}{2}\Omega_{\pi,0}\left(1+t-ig(z)\right) & \mathrm{\!\!\! if}\: U=ti+g(z)\,\wedge\, t<-1,\\
\arctan\left(c\left(z\right)\right) & \mathrm{\!\!\! otherwise}.\end{cases}\end{equation}
Formula (\ref{eq:atanViaAtanh}) should automatically achieve these
simplified forms except probably for the three cases where $\alpha<0$.

Table \ref{tab:SeriesAtan} lists some simplified correct 0-degree
terms for $\arctan\left(\ldots\right)$. Corresponding examples for
$\mathrm{arctanh}\left(\ldots\right)$ can be derived from (\ref{eq:atanViaAtanh}).

\begin{table}[h]
\caption{$\mbox{0-degree term of series}\left(\arctan\left(u\right),\, z\!=\!0,\, o\left(z^{n}\right)\right)$
with $z\in\mathbb{C},\: x\in\mathbb{R},\: n\geq3$ : \label{tab:SeriesAtan}}

\begin{tabular}{|c|r|l|}
\hline 
\# & $u$ & A correct 0-degree term near $z=0$\tabularnewline
\hline
\hline 
1 & $z^{-2}+z^{-1}$ & $\dfrac{\pi}{2}\begin{cases}
\pm1 & \mathrm{if}\: z=0\\
1 & \mathrm{if}\:-\pi/2<\arg\left(z^{-2}+z^{-1}\right)\leq\pi/2\\
-1 & \mathrm{otherwise}\end{cases}$\tabularnewline
\hline 
2 & $2i+ze^{z}$ & $\dfrac{\ln3}{2}i+\begin{cases}
\dfrac{\pi}{2} & \mathrm{if}\:\Re\left(z+z^{2}+\dfrac{z^{3}}{6}+\cdots+o\left(z^{n}\right)\right)\leq0\\
-\dfrac{\pi}{2} & \mathrm{otherwise}\end{cases}$\tabularnewline
\hline 
3 & $2i+z^{1/4}e^{z}$ & $\dfrac{\pi}{2}+\dfrac{\ln3}{2}i$\tabularnewline
\hline 
4 & $-2i+ze^{z}$ & $\dfrac{\ln3}{2}i+\begin{cases}
\dfrac{\pi}{2} & \mathrm{if}\: z\neq0\\
-\dfrac{\pi}{2} & \mathrm{otherwise}\end{cases}$\tabularnewline
\hline 
5 & $-2i-z^{1/4}e^{z}$ & $-\dfrac{\pi}{2}-\dfrac{\ln3}{2}i$\tabularnewline
\hline 
6 & $i+ze^{z}$ & $\dfrac{\pi-\ln\left(iz\right)+\ln2}{2}i-\begin{cases}
0 & \mathrm{\!\!\! if}\:\arg\!\left(1\!+\! z\!+\!\dfrac{z^{2}}{2}\!+\!\dfrac{z^{3}}{6}\!+\cdots+o\!\left(z^{n}\right)\right)+\arg\!\left(iz\right)\leq\pi\\
\pi & \mathrm{\!\!\! otherwise}\end{cases}$\tabularnewline
\hline 
7 & $i+ize^{z}$ & $\ln\left(\dfrac{z}{2}\right)-\dfrac{i}{2}$\tabularnewline
\hline 
8 & $-i+z^{2}e^{z}$ & $\dfrac{i\ln\left(\dfrac{-iz^{2}}{2}\right)}{2}+\begin{cases}
\pi & \mathrm{\!\! if}\:\arg\!\left(1\!+\! z\!+\!\dfrac{z^{2}}{2}\!+\!\dfrac{z^{3}}{6}\!+\cdots+o\!\left(z^{n}\right)\right)+\arg\left(-iz^{2}\right)\leq-\pi\\
\pi & \mathrm{\!\! if}\:\arg\!\left(1\!+\! z\!+\!\dfrac{z^{2}}{2}\!+\!\dfrac{z^{3}}{6}\!+\cdots+o\!\left(z^{n}\right)\right)+\arg\left(-iz^{2}\right)>\pi\\
0 & \mathrm{\!\! otherwise}\end{cases}$\tabularnewline
\hline 
9 & $-i+iz^{1/2}e^{z}$ & $\dfrac{i\ln\left(z/4\right)}{4}$\tabularnewline
\hline 
10 & $x^{-2}+x^{-1}$ & $\dfrac{\pi}{2}\begin{cases}
\pm1 & \mathrm{if}\: x=0\\
1 & \mathrm{if}\: x>0\\
-1 & \mathrm{otherwise}\end{cases}$\tabularnewline
\hline 
11 & $2i+x^{3/4}e^{x}$ & $\dfrac{\ln3}{2}i+\begin{cases}
\dfrac{\pi}{2} & \mathrm{if}\: x\geq0\\
-\dfrac{\pi}{2} & \mathrm{otherwise}\end{cases}$\tabularnewline
\hline 
12 & $2i+x^{1/2}e^{x}$ & $\pi/2+i\left(\ln3\right)/2$\tabularnewline
\hline 
13 & $-2i+x^{3/4}e^{x}$ & $-\dfrac{\ln3}{2}i+\begin{cases}
\dfrac{\pi}{2} & \mathrm{if}\: x>0\\
-\dfrac{\pi}{2} & \mathrm{otherwise}\end{cases}$\tabularnewline
\hline 
14 & $-2i-x^{3/4}e^{x}$ & $-\pi/2-i\left(\ln3\right)/2$\tabularnewline
\hline 
15 & $i+x^{5/4}e^{x}$ & $\dfrac{-\ln\left(ix^{5/4}/2\right)}{2}$\tabularnewline
\hline 
16 & $-i+x^{5/4}e^{x}$ & $\dfrac{\ln\left(-ix^{5/4}/2\right)}{2}$\tabularnewline
\hline
\end{tabular}
\end{table}

\subsection{Series for arcsinh, arcsin and arccos\label{sub:Series-for-arcsinh-etc}}

The inverse arcsine, arccosine and hyperbolic arcsine of a series
$U$ can be computed from\begin{eqnarray}
\mbox{arcsinh}\left(U\right) & \equiv & \ln\left(U+\sqrt{1+U^{2}}\right),\label{eq:asinhViaLn}\\
\arcsin\left(w\right) & \equiv & -i\,\mathrm{arcsinh}\left(iw\right),\label{eq:arcsinViaArcsinh}\\
\arccos\left(w\right) & \equiv & \dfrac{\pi}{2}-\arcsin\left(w\right).\label{eq:ArcCosViaArcSin}\end{eqnarray}
The square root might contribute a factor containing expressions involving
$\arg\left(\ldots\right)$ to all of its result terms. The subsequent
logarithm in identity (\ref{eq:asinhViaLn}) might then contribute
horrid nested instances of $\arg\left(\ldots\right)$ to a 0-degree
term in the result.

\textsl{Mathematica} 7.0.1.0 avoids this by instead using correction
terms and factors that are specific to inverse trigonometric and inverse
hyperbolic functions: The resulting arguments of $\arg\left(\ldots\right)$
are unnested and entail the terms of $U$, which are often fewer and
simpler than those of $1+U^{2}$ and $U+\sqrt{1+U^{2}}$. Extensive
numeric experiments provide convincing evidence that they are correct
in most instances. Most of the formulas in this subsection and the
following one are simplified and corrected versions of those corrections.
When using these corrections and identity (\ref{eq:asinhViaLn}),
it is important to use the rewrite $\ln\left(pz^{\alpha}\right)\rightarrow\ln(p)+\alpha\ln(z)$,
but suppress the branch corrections for logarithms and fractional
powers discussed in previous sections. It is also important to compute
intermediate series such as $\sqrt{1+U^{2}}$ to sufficient order
so that the final order is as requested.

A correction term or factor is unnecessary for $\mathrm{arcsinh}\left(U(z)\right)$
if the dominant degree of $U(z)$ is positive or if the dominant degree
is 0 and $U\left(0\right)$ isn't on either branch cut or either branch
point. Otherwise, here are the different cases:

\subsubsection{\label{sub:AsinhDeg0CfLtMinus1OrGt1TimesI}Case $\mathrm{arcsinh}\left(U(z)=\hat{c}i+g(z)\right)$
with $\hat{c}<-1\,\vee\,\hat{c}>1$}

If the dominant term of $U(z)$ is $\hat{c}i$ with $\hat{c}<-1\,\vee\,\hat{c}>1$,
then\begin{eqnarray*}
\mathrm{arcsinh}\left(\hat{c}i+g\left(z\right)\right) & = & \dfrac{i\pi\mathrm{\, sign}(\hat{c})}{2}+T\cdot\left(\mathrm{arccosh}\left(\hat{c}\right)+o\left(z^{0}\right)\right),\end{eqnarray*}
 where\begin{eqnarray}
T & = & \begin{cases}
1 & \mathrm{if\;}\mathrm{sign}(\hat{c})\,\Re\left(\, g(z)\right)\geq0,\\
-1 & \mathrm{otherwise.}\end{cases}\end{eqnarray}

\subsubsection{\label{sub:AsinhDeg0CfEqIOrMinusI}Case $\mathrm{arcsinh}\left(U(z)=\bar{c}i+g(z)\right)$
with $\overline{c}=1\,\vee\,\bar{c}=-1$}

If the dominant term of $U(z)$ is $\bar{c}i$ with $\overline{c}=1\,\vee\,\bar{c}=-1$
and the next non-zero term is $bz^{\sigma}$, then\begin{eqnarray*}
\mathrm{arcsinh}\left(\bar{c}i\!+\! bz^{\sigma}\!+\! h\!\left(z\right)\right) & = & \dfrac{i\pi\bar{c}}{2}+(-1)^{P}Q\cdot\left(i\bar{c}\sqrt{2ib}z^{\sigma/2}+o\left(z^{\sigma/2}\right)\right)\end{eqnarray*}
where\begin{eqnarray}
P & = & \left\lfloor \dfrac{1}{2}-\dfrac{\arg\left(i\bar{c}b+\dfrac{i\bar{c}h(z)}{z^{\sigma}}\right)+\sigma\arg(z)}{2\pi}\right\rfloor ,\\
Q & = & \begin{cases}
-1 & \mathrm{if\:}\arg\left(b\right)=\dfrac{\pi\bar{c}}{2}\wedge\Re\left(\dfrac{h(z)}{z^{\sigma}}\right)<0,\\
1 & \mathrm{otherwise}.\end{cases}\end{eqnarray}

Expression $i\bar{c}b+\dfrac{i\bar{c}h(z)}{z^{\sigma}}\rightarrow i\bar{c}b$
as $z\rightarrow0$, so the critical angles for $P$ are\begin{equation}
-\pi\,<\,\theta_{c}=\dfrac{\left(2n+1\right)\pi-\arg\left(i\bar{c}b\right)}{\sigma}\,\leq\,\pi\end{equation}
for all integer $n$ satisfying\begin{equation}
\left\lfloor \dfrac{\arg\left(i\bar{c}b\right)}{2\pi}-\dfrac{1}{2}-\dfrac{\sigma}{2}\right\rfloor <n\leq\left\lfloor \dfrac{\arg\left(i\bar{c}b\right)}{2\pi}-\dfrac{1}{2}+\dfrac{\sigma}{2}\right\rfloor .\end{equation}
There are no such angles if $0<\sigma<1$ and $\arg\left(i\bar{c}b\right)$
is sufficiently close to 0: If \begin{equation}
\left(\sigma-1\right)\pi\,<\,\arg\left(i\bar{c}b\right)\,<\,\left(1-\sigma\right)\pi,\end{equation}
then $P\equiv0$, which is easily tested.

If $0<\sigma\leq2$, then a more candid and accurate representation
of $(-1)^{P}$ is\begin{equation}
\begin{cases}
1 & \mathrm{if}\:-\pi<\arg\left(i\bar{c}b+\dfrac{i\bar{c}h(z)}{z^{\sigma}}\right)+\sigma\arg(z)\leq\pi,\\
-1 & \mathrm{otherwise}.\end{cases}\end{equation}

\subsubsection{\label{sub:AsinhNegDegImCf}Case $\mathrm{arcsinh\left(\ldots\right)}$
with negative dominant exponent and imaginary dominant coefficient}

Use formula (\ref{eq:asinhViaLn}) including the angle rotation factor
for the square root and the correction term for the logarithm. Brace
yourself for a truly ugly result.

\subsubsection{\label{sub:AsinhNegDegCfNotIm}Case $\mathrm{arcsinh\left(\ldots\right)}$
with negative dominant exponent and non-imaginary dominant coefficient}

If the dominant exponent $\alpha$ of $U(z)$ is negative and the
dominant coefficient $c$ isn't pure imaginary, then\begin{eqnarray}
\mathrm{arcsinh}\left(cz^{\alpha}+g(z)\right) & = & (-1)^{W}N\cdot\left(i\pi W+\dfrac{\ln\left(4c^{2}\right)}{2}+\alpha\ln\left(z\right)+o\left(z^{0}\right)\right).\end{eqnarray}
where\begin{eqnarray}
W & = & \left\lfloor \dfrac{1}{2}-\dfrac{\arg\left(\left(c+\dfrac{g(z)}{z^{\alpha}}\right)^{2}\right)+2\alpha\arg(z)}{2\pi}\right\rfloor ,\\
N & = & \begin{cases}
1 & \mathrm{if\:}\Re(c)>0,\\
-1 & \mathrm{otherwise}.\end{cases}\end{eqnarray}

Expression $\left(c+\dfrac{g(z)}{z^{\alpha}}\right)^{2}\rightarrow c^{2}$
as $z\rightarrow0$, so the critical angles for $W$ are\begin{equation}
-\pi\,<\,\theta_{c}=\dfrac{\left(2n+1\right)\pi-\arg\left(c^{2}\right)}{2\alpha}\,\leq\,\pi\end{equation}
for all integer $n$ satisfying\begin{equation}
\left\lfloor \dfrac{\arg\left(c^{2}\right)}{2\pi}-\dfrac{1}{2}-\dfrac{\alpha}{2}\right\rfloor <n\leq\left\lfloor \dfrac{\arg\left(c^{2}\right)}{2\pi}-\dfrac{1}{2}+\dfrac{\alpha}{2}\right\rfloor .\end{equation}
There are no such angles if $-1/2<\alpha<0$ and $\arg\left(c^{2}\right)$
is sufficiently close to 0: If\begin{equation}
\left(\alpha-1\right)\pi\,<\,\arg\left(c^{2}\right)\,<\,\left(1-\alpha\right)\pi,\end{equation}
then $W\equiv0$, which is easily tested.

If $-1\leq\alpha<0$, then more candidly and accurately, \begin{eqnarray}
W & = & \begin{cases}
1 & \mathrm{if}\:\arg\left(\left(c+\dfrac{g(z)}{z^{\alpha}}\right)^{2}\right)+2\alpha\arg(z)>\pi,\\
-1 & \mathrm{if}\:\arg\left(\left(c+\dfrac{g(z)}{z^{\alpha}}\right)^{2}\right)+2\alpha\arg(z)\leq-\pi,\\
0 & \mathrm{otherwise};\end{cases}\\
(-1)^{W} & = & \begin{cases}
1 & \mathrm{if}\:-\pi<\arg\left(\left(c+\dfrac{g(z)}{z^{\alpha}}\right)^{2}\right)+2\alpha\arg(z)\leq\pi,\\
-1 & \mathrm{otherwise}.\end{cases}\end{eqnarray}

Some incorrect series for $\mathrm{arcsinh}\left(\ldots\right)$,
$\arcsin\left(\ldots\right)$ and $\arccos\left(\ldots\right)$ in
current implementations are attributable to incorrect or omitted correction
terms and/or factors. Table \ref{tab:SeriesAsinh} lists some examples
for $\mathrm{arcsinh}\left(\ldots\right)$. Corresponding examples
for $\arcsin\left(\ldots\right)$ and $\arccos\left(\ldots\right)$
can be derived from (\ref{eq:arcsinViaArcsinh}) and (\ref{eq:ArcCosViaArcSin}).

\begin{table}[h]
\caption{$\mbox{Fir\mbox{st} two non-zero terms of series}\left(\mathrm{arcsinh}\left(u\right),\, z\!=\!0,\, o\left(z^{\infty}\right)\right)$
with $z\in\mathbb{C}$, $x\in\mathbb{R}$: \label{tab:SeriesAsinh}}

\begin{tabular}{|c|r|l|l|}
\hline 
\negthinspace{}\negthinspace{}\#\negthinspace{}\negthinspace{} & $u$ & First two non-0 terms of $\mathrm{arcsinh\,}u$ near $\mathrm{variable}=0$ & Sec.\tabularnewline
\hline
\hline 
\negthinspace{}\negthinspace{}1\negthinspace{}\negthinspace{} & -$2i+z^{2}+z^{3}\!$ & $-\dfrac{i\pi}{2}+\left(\begin{cases}
1 & \mathrm{if}\:\Re\left(z^{2}+z^{3}\right)>0\\
-1 & \mathrm{otherwise}\end{cases}\right)\left(\ln\left(2+\sqrt{3}\right)+\dfrac{iz^{2}}{\sqrt{3}}+\cdots\right)$ & \negthinspace{}\ref{sub:AsinhDeg0CfLtMinus1OrGt1TimesI}\negthinspace{}\tabularnewline
\hline 
\negthinspace{}\negthinspace{}2\negthinspace{}\negthinspace{} & $2i+z^{2}+z^{3}$ & $\dfrac{i\pi}{2}+\left(\begin{cases}
1 & \mathrm{if}\:\Re\left(z^{2}+z^{3}\right)\geq0\\
-1 & \mathrm{otherwise}\end{cases}\right)\left(-\dfrac{\ln\left(2+\sqrt{3}\right)}{2}-\dfrac{iz^{2}}{\sqrt{3}}+\cdots\right)$ & \negthinspace{}\ref{sub:AsinhDeg0CfLtMinus1OrGt1TimesI}\negthinspace{}\tabularnewline
\hline 
\negthinspace{}\negthinspace{}3\negthinspace{}\negthinspace{} & $2i+z^{1/4}$ & $\dfrac{i\pi}{2}+\ln\left(2+\sqrt{3}\right)-\dfrac{iz^{1/4}}{\sqrt{3}}+\cdots$ & \negthinspace{}\ref{sub:AsinhDeg0CfLtMinus1OrGt1TimesI}\negthinspace{}\tabularnewline
\hline 
\negthinspace{}\negthinspace{}4\negthinspace{}\negthinspace{} & $i+z^{2}+z^{3}$ & $\dfrac{i\pi}{2}+\left(-1\right)^{\left\lfloor \dfrac{1}{2}\,-\,\dfrac{\arg\left(-1+iz\right)+2\arg z}{2\pi}\right\rfloor }\left(iz+\cdots\right)$ & \negthinspace{}\ref{sub:AsinhDeg0CfEqIOrMinusI}\negthinspace{}\tabularnewline
\hline 
\negthinspace{}\negthinspace{}5\negthinspace{}\negthinspace{} & $i+iz^{2}+iz^{3}$ & $\!\dfrac{i\pi}{2}\!+\!\left(-1\right)^{\left\lfloor \dfrac{1}{2}\,-\,\dfrac{\arg\left(-1\!+\! iz\right)\!+\!2\arg z}{2\pi}\right\rfloor }\!\left(\!\begin{cases}
1 & \!\!\!\Re\!\left(z\right)\!\geq\!0\\
-1 & \mathrm{\!\!\! otherwise}\end{cases}\right)\!\left(\sqrt{2}z\!+\cdots\right)\!$ & \negthinspace{}\ref{sub:AsinhDeg0CfEqIOrMinusI}\negthinspace{}\tabularnewline
\hline 
\negthinspace{}\negthinspace{}6\negthinspace{}\negthinspace{} & $\!-i+iz^{2}+iz^{3}\!$ & $-\dfrac{i\pi}{2}+\left(-1\right)^{\left\lfloor \dfrac{1}{2}\,-\,\dfrac{\arg\left(-1+iz\right)+2\arg z}{2\pi}\right\rfloor }\left(i\sqrt{2}z+\cdots\right)$ & \negthinspace{}\ref{sub:AsinhDeg0CfEqIOrMinusI}\negthinspace{}\tabularnewline
\hline 
\negthinspace{}\negthinspace{}7\negthinspace{}\negthinspace{} & $-i-iz^{2}+z^{3}$ & $\!-\dfrac{i\pi}{2}\!+\!\left(-1\right)^{\left\lfloor \dfrac{1}{2}\,-\,\dfrac{\arg\left(-1\!+\! iz\right)\!+\!2\arg z}{2\pi}\right\rfloor }\!\left(\!\begin{cases}
1, & \mathrm{\!\!\!}\Re\!\left(z\right)\!\leq\!0\\
-1 & \mathrm{\!\!\! otherwise}\end{cases}\right)\!\left(-\sqrt{2}z\!+\cdots\right)\!$ & \negthinspace{}\ref{sub:AsinhDeg0CfEqIOrMinusI}\negthinspace{}\tabularnewline
\hline 
\negthinspace{}\negthinspace{}8\negthinspace{}\negthinspace{} & $iz^{-1}+1$ & ugh! & \negthinspace{}\ref{sub:AsinhNegDegImCf}\negthinspace{}\tabularnewline
\hline 
\negthinspace{}\negthinspace{}9\negthinspace{}\negthinspace{} & $z^{-2}+z^{-1}$ & $\!(-1)^{w}\left(\ln2\!-\!2\ln z\!+\!\pi iw\!+\! z\!+\!\cdots\right)|\, w\!=\!\left\lfloor \dfrac{1}{2}-\dfrac{\arg\left(1+2z+z^{2}\right)\!+\!4\arg z}{2\pi}\right\rfloor \!$ & \negthinspace{}\ref{sub:AsinhNegDegCfNotIm}\negthinspace{}\tabularnewline
\hline 
\negthinspace{}\negthinspace{}10\negthinspace{}\negthinspace{} & $-2i+x^{3/2}$ & $\left(-\dfrac{\pi i}{2}-\ln\left(2+\sqrt{3}\right)\begin{cases}
1 & \mathrm{if}\: x\leq0\\
-1 & \mathrm{otherwise}\end{cases}\right)-\dfrac{ix^{3/2}}{\sqrt{3}}+\cdots$ & \negthinspace{}\ref{sub:AsinhDeg0CfLtMinus1OrGt1TimesI}\negthinspace{}\tabularnewline
\hline 
\negthinspace{}\negthinspace{}11\negthinspace{}\negthinspace{} & $-2i-x^{3/2}$ & $\left(-\dfrac{\pi i}{2}-\ln\left(2+\sqrt{3}\right)\right)+\dfrac{ix^{3/2}}{\sqrt{3}}+\cdots$ & \negthinspace{}\ref{sub:AsinhDeg0CfLtMinus1OrGt1TimesI}\negthinspace{}\tabularnewline
\hline 
\negthinspace{}\negthinspace{}12\negthinspace{}\negthinspace{} & $2i+x^{3/4}$ & $\left(\dfrac{\pi i}{2}+\ln\left(2+\sqrt{3}\right)\begin{cases}
1 & \mathrm{if}\: x\geq0\\
-1 & \mathrm{otherwise}\end{cases}\right)-\dfrac{ix^{3/4}}{\sqrt{3}}+\cdots$ & \negthinspace{}\ref{sub:AsinhDeg0CfLtMinus1OrGt1TimesI}\negthinspace{}\tabularnewline
\hline 
\negthinspace{}\negthinspace{}13\negthinspace{}\negthinspace{} & $2i+x^{3/2}$ & $\left(\dfrac{\pi i}{2}+\ln\left(2+\sqrt{3}\right)\right)-\dfrac{ix^{3/2}}{\sqrt{3}}+\cdots$ & \negthinspace{}\ref{sub:AsinhDeg0CfLtMinus1OrGt1TimesI}\negthinspace{}\tabularnewline
\hline 
\negthinspace{}\negthinspace{}14\negthinspace{}\negthinspace{} & $i+x^{3/2}$ & $\dfrac{\pi i}{2}+\left(\begin{cases}
1 & \mathrm{if}\: x\geq0\\
-1 & \mathrm{otherwise}\end{cases}\right)\left(\left(1-i\right)x^{3/4}+\cdots\right)$ & \negthinspace{}\ref{sub:AsinhDeg0CfEqIOrMinusI}\negthinspace{}\tabularnewline
\hline 
\negthinspace{}\negthinspace{}15\negthinspace{}\negthinspace{} & $i-x^{3/2}$ & $\dfrac{\pi i}{2}-\left(1+i\right)x^{3/4}+\cdots$ & \negthinspace{}\ref{sub:AsinhDeg0CfEqIOrMinusI}\negthinspace{}\tabularnewline
\hline 
\negthinspace{}\negthinspace{}16\negthinspace{}\negthinspace{} & $-i+x^{5/2}$ & $-\dfrac{\pi i}{2}+\left(\begin{cases}
1 & \mathrm{if}\: x\geq0\\
-1 & \mathrm{otherwise}\end{cases}\right)\left(\left(1+i\right)x^{5/4}+\cdots\right)$ & \negthinspace{}\ref{sub:AsinhDeg0CfEqIOrMinusI}\negthinspace{}\tabularnewline
\hline 
\negthinspace{}\negthinspace{}17\negthinspace{}\negthinspace{} & $-i+x^{3/2}$ & $-\dfrac{\pi i}{2}+\left(1+i\right)x^{3/4}+\cdots$ & \negthinspace{}\ref{sub:AsinhDeg0CfEqIOrMinusI}\negthinspace{}\tabularnewline
\hline 
\negthinspace{}\negthinspace{}18\negthinspace{}\negthinspace{} & $x^{-2}+x^{-1}$ & $\left(\ln2-2\ln x+\begin{cases}
0 & \mathrm{if}\: x\geq0\\
2\pi i & \mathrm{otherwise}\end{cases}\right)+x+\cdots$ & \negthinspace{}\ref{sub:AsinhNegDegCfNotIm}\negthinspace{}\tabularnewline
\hline
\end{tabular}
\end{table}

\subsection{Branch corrections for arccosh}

Subsection \ref{sub:Series-for-arcsinh-etc} expressed the series
for arcsin and arccos but not arccosh in terms of the series for arcsinh.
The inverse hyperbolic arccosine of a series $U$ can instead be computed
from either of\begin{eqnarray}
\mbox{arccosh}\left(U\right) & \equiv & 2\ln\left(\sqrt{\dfrac{U-1}{2}}+\sqrt{\dfrac{U+1}{2}}\right),\label{eq:acoshViaLn1}\\
\mbox{arccosh} & \equiv & \ln\left(U+\sqrt{U-1}\sqrt{U+1}\right).\label{eq:acoshViaLn2}\end{eqnarray}

The latter is slightly slower asymptotically because of the series
multiplication. The alternative that gives a simpler result might
depend on the general expression simplifier and the case, such as
whether or not the dominant degree of $U$ is negative.

Each square root in these formulas might contribute a factor containing
expressions involving $\arg\left(\ldots\right)$ to all of its result
terms. The subsequent logarithm in identity (\ref{eq:acoshViaLn1})
or (\ref{eq:acoshViaLn2}) might then contribute horrid nested instances
of $\arg\left(\ldots\right)$ to a 0-degree term in the result. Here
are some simplified correction terms and factors that avoid this:

Let the dominant term of $U$ be $cz^{\alpha}$, and let the sum of
the remaining terms be $g(z)$. Let $bz^{\sigma}$ be the dominant
term of $g(z),$ and let $h(z)$ be the remaining terms of $g(z)$.
A correction term and/or factor is unnecessary if $\alpha=0\,\wedge\,\left(\Im\left(c\right)\neq0\,\vee\, c>1\right)$.
Otherwise we have the following cases:

\subsubsection{\label{sub:AcoshDeg0CfLtMinus1}Case $\mathrm{arccosh}\left(U(z)\right)$
with dominant deg 0 \& dominant coefficient < -1}

If $\alpha=0$ and $c<-1$, then\begin{eqnarray}
\mathrm{arccosh}\left(c(z)+g\left(z\right)\right) & = & 2i\pi J+\mathrm{arccosh}\left(c(z)\right)+o\left(z^{0}\right)\end{eqnarray}
where\begin{eqnarray}
J & = & \begin{cases}
-1 & \mathrm{if\:}\Im\left(g(z)\right)<0,\\
0 & \mathrm{otherwise}.\end{cases}\end{eqnarray}

\subsubsection{\label{sub:AcoshDeg0CfInMinus1ZeroOrZero1}Case $\mathrm{arccosh}\left(U(z)\right)$
with dominant degree 0 and dominant coefficient in (-1, 0) or (0,1)}

If $\alpha=0\,\wedge\,\left(-1<c<0\:\vee\:0<c<1\right)$ then \begin{equation}
\mathrm{arccosh}\left(c+g\left(z\right)\right)=K\cdot\left(\mathrm{arccosh}\left(c\right)+o\left(z^{0}\right)\right),\end{equation}
where\begin{equation}
K\,=\,(-1)^{J}\,=\,\begin{cases}
-1 & \mathrm{if\:}\Im\left(g(z)\right)<0,\\
1 & \mathrm{otherwise.}\end{cases}\end{equation}

\subsubsection{\label{sub:AcoshhDeg0Cf1}Case $\mathrm{arccosh}\left(U(z)\right)$
with dominant degree 0 and dominant coefficient 1}

If instead $\alpha=0$ and $c=1$, then\begin{eqnarray}
\mathrm{arccosh}\left(1+bz^{\sigma}+h\left(z\right)\right) & = & BE\cdot\left(\sqrt{2b}z^{\sigma/2}+o\left(z^{\sigma/2}\right)\right),\end{eqnarray}
where\begin{eqnarray}
B & = & (-1)^{\left\lfloor \dfrac{1}{2}\,-\,\dfrac{\arg\left(b+\dfrac{h(z)}{z^{\sigma}}\right)+\sigma\arg\left(z\right)}{2\pi}\right\rfloor },\\
E & = & \begin{cases}
-1 & \mathrm{if\:}\arg\left(b\right)=\pi\:\wedge\:\Im\left(\dfrac{h\left(z\right)}{z^{\sigma}}\right)<0,\\
1 & \mathrm{otherwise}.\end{cases}\end{eqnarray}

Expression $b+\frac{h(z)}{z^{\sigma}}\rightarrow b$ as $z\rightarrow0$,
so the critical angles for $B$ are\begin{equation}
-\pi\,<\,\theta_{c}=\dfrac{\left(2n+1\right)\pi-\arg\left(b\right)}{\sigma}\,\leq\,\pi\end{equation}
for all integer $n$ satisfying\begin{equation}
\left\lfloor \dfrac{\arg\left(b\right)}{2\pi}-\dfrac{1}{2}-\dfrac{\sigma}{2}\right\rfloor <n\leq\left\lfloor \dfrac{\arg\left(b\right)}{2\pi}-\dfrac{1}{2}+\dfrac{\sigma}{2}\right\rfloor .\end{equation}
There are no such angles if $0<\sigma<1$ and $\arg\left(b\right)$
is sufficiently close to 0: If\begin{equation}
\left(\sigma-1\right)\pi\,<\,\arg\left(b\right)\,<\,\left(1-\sigma\right)\pi,\end{equation}
then $B\equiv0$, which is easily tested.

If $0<\sigma\leq2$, then a more candid and accurate representation
of $B$ is\begin{equation}
\begin{cases}
1 & \mathrm{if}\:-\pi<\arg\left(b+\dfrac{h(z)}{z^{\sigma}}\right)+\sigma\arg\left(z\right)\leq\pi,\\
-1 & \mathrm{otherwise}.\end{cases}\end{equation}

\subsubsection{\label{sub:AcoshDeg0CfMinus1}Case $\mathrm{arccosh}\left(U(z)\right)$
with dominant degree 0 and dominant coefficient -1}

If $\alpha=0$ and $c=-1$, then\begin{eqnarray*}
\mathrm{arccosh}\left(-1+bz^{\sigma}+h\left(z\right)\right) & = & i\pi C+CBE\cdot\left(i\sqrt{2b}z^{\sigma/2}+o\left(z^{\sigma/2}\right)\right)\end{eqnarray*}
where\begin{eqnarray}
C & = & \begin{cases}
1 & \mathrm{if\:}\Im\left(bz^{\sigma}+h\left(z\right)\right)<0,\\
-1 & \mathrm{otherwise}.\end{cases}\end{eqnarray}

\subsubsection{\label{sub:AcoshNegDeg}Case $\mathrm{arccosh}\left(U(z)\right)$
with negative dominant degree}

If instead $\alpha<0$ then\begin{eqnarray}
\mathrm{arccosh}\left(cz^{\alpha}+g\left(z\right)\right) & = & 2i\pi\left(D_{1}+M\right)+\ln\left(2c\right)+\alpha\ln\left(z\right)+o\left(z^{0}\right)\end{eqnarray}
where\begin{eqnarray}
D_{1} & = & \left\lfloor \dfrac{1}{2}-\dfrac{\arg\left(c+\dfrac{g\left(z\right)}{z^{\alpha}}\right)+\alpha\arg\left(z\right)}{2\pi}\right\rfloor ,\\
M & = & \begin{cases}
1 & \mathrm{if\:}\arg\left(c\right)=\pi\:\wedge\:\Im\left(\dfrac{g(z)}{z^{\alpha}}\right)<0,\\
0 & \mathrm{otherwise}.\end{cases}\end{eqnarray}

Expression $c+\frac{g\left(z\right)}{z^{\alpha}}\rightarrow c$ as
$z\rightarrow0$, so the critical angles for $D_{1}$ are\begin{equation}
-\pi\,<\,\theta_{c}=\dfrac{\left(2n+1\right)\pi-\arg\left(c\right)}{\alpha}\,\leq\,\pi\end{equation}
for all integer $n$ satisfying\begin{equation}
\left\lfloor \dfrac{\arg\left(c\right)}{2\pi}-\dfrac{1}{2}-\dfrac{\alpha}{2}\right\rfloor <n\leq\left\lfloor \dfrac{\arg\left(c\right)}{2\pi}-\dfrac{1}{2}+\dfrac{\alpha}{2}\right\rfloor .\end{equation}
There are no such angles if $-1<\alpha<0$ and $\arg\left(c\right)$
is sufficiently close to 0: If\begin{equation}
-\left(1+\alpha\right)\pi\,<\,\arg\left(c\right)\,<\,\left(1+\alpha\right)\pi,\end{equation}
then $D_{1}\equiv0$, which is easily tested.

If $-2<\alpha<0$, then a more candid and accurate representation
of $D_{1}$ is\begin{equation}
\begin{cases}
1 & \mathrm{if}\:-\pi<\arg\left(c+\dfrac{g(z)}{z^{\alpha}}\right)+\alpha\arg\left(z\right)\leq\pi,\\
-1 & \mathrm{otherwise}.\end{cases}\end{equation}

\subsubsection{\label{sub:AcoshPositiveDeg}Case $\mathrm{arccosh}\left(U(z)\right)$
with positive dominant degree}

If instead $\alpha>0$, then\begin{eqnarray}
\mathrm{arccosh}\left(cz^{\alpha}+g\left(z\right)\right) & = & (-1)^{D_{1}}\left(-1\right)^{D_{2}}G\cdot\left(\dfrac{i\pi}{2}+o\left(z^{0}\right)\right)\end{eqnarray}
where\begin{eqnarray}
D_{2} & = & \left\lfloor \dfrac{1}{2}-\dfrac{1}{2\pi}\left(\arg\left(\dfrac{-1}{c+\dfrac{g(z)}{z^{\alpha}}}\right)-\alpha\arg(z)\right)\right\rfloor ,\label{eq:DefineD2}\\
G & = & \begin{cases}
-1 & \mathrm{if\:}z\neq0\wedge\Im\left(c\right)<0\:\vee\\
 & \quad\arg\left(c\right)=\pi\wedge\Im\left(\dfrac{g(z)}{z^{\alpha}}\right)<0\:\vee\\
 & \quad\arg\left(c\right)=0\wedge\Im\left(\dfrac{g(z)}{cz^{\alpha}+g(z)}\right)<0,\\
1 & \mathrm{otherwise}.\end{cases}\label{eq:DefineG}\end{eqnarray}

Expression $-1/\left(c+g(z)/z^{\alpha}\right)\rightarrow-1/c$ as
$z\rightarrow0$, so the critical angles for $D_{2}$ are\begin{equation}
-\pi\,<\,\theta_{c}=\dfrac{\left(2n+1\right)\pi-\arg\left(-1/c\right)}{\alpha}\,\leq\,\pi\end{equation}
for all integer $n$ satisfying\begin{equation}
\left\lfloor \dfrac{\arg\left(-1/c\right)}{2\pi}-\dfrac{1}{2}-\dfrac{\alpha}{2}\right\rfloor <n\leq\left\lfloor \dfrac{\arg\left(-1/c\right)}{2\pi}-\dfrac{1}{2}+\dfrac{\alpha}{2}\right\rfloor .\end{equation}
There are no such angles if $0<\alpha<1$ and $\arg\left(c\right)$
is sufficiently close to 0: If\begin{equation}
\left(\alpha-1\right)\pi\,<\,\arg\left(-1/c\right)\,<\,\left(1-\alpha\right)\pi,\end{equation}
then $D_{2}\equiv0$, which is easily tested.

If $0<\alpha\leq2$, then a more candid and accurate representation
of $D_{2}$ is\begin{equation}
\begin{cases}
1 & \mathrm{if}\:-\pi<\arg\left(-1/\left(c+\dfrac{g(z)}{z^{\alpha}}\right)\right)+\alpha\arg\left(z\right)\leq\pi,\\
-1 & \mathrm{otherwise}.\end{cases}\end{equation}

When using these corrections and identity (\ref{eq:acoshViaLn1})
or (\ref{eq:acoshViaLn2}), it is important to use the rewrite $\ln\left(cz^{\alpha}\right)\rightarrow\ln(c)+\alpha\ln(z)$
but suppress in these formulas the logarithmic and fractional power
adjustments discussed in Sections \ref{sec:Branch-bugs-for-ln} and
\ref{sec:BugsForFracPows}. Also, sub-expressions such as\[
\dfrac{-1}{c+\dfrac{g(z)}{z^{\alpha}}}\quad\mathrm{and}\quad\dfrac{g(z)}{cz^{\alpha}+g(z)}\]
in (\ref{eq:DefineD2}) and (\ref{eq:DefineG}) can and should be
approximated by appropriate-order series. Note that a truncated series
for these two sub-expressions usually can't be exact for non-zero
$g(z)$. Therefore we should expect some narrow cusps in which the
corrections are not piecewise constant. If $g(z)$ is known to $o\left(z^{m}\right)$,
then we can compute the series for these two sub-expressions to $o\left(z^{m-\alpha}\right)$.

For both $\mathrm{arcsinh}\left(U(z)\right)$ and $\mathrm{arccosh}\left(U(z)\right)$,
if we are unable to narrow the choice to one of the above cases for
reasons such as $c$ containing insufficiently constrained indeterminates
other than $z$, then we should construct a piecewise result from
all of the cases that we can't preclude.

Table \ref{tab:SeriesACosh} contains some correct results for $\mathrm{arccosh}\left(\ldots\right)$
series.

\begin{table}[h]
\caption{$\mbox{\ensuremath{1^{\mathrm{st}}} two non-0 terms of series}\left(\mathrm{arccosh}\left(u\right),\, z\!=\!0,\, o\left(z^{\infty}\right)\right)$
with $z\in\mathbb{C}$, $x\in\mathbb{R}$: \label{tab:SeriesACosh}}

\begin{tabular}{|c|r|l|l|}
\hline 
\negthinspace{}\negthinspace{}\#\negthinspace{}\negthinspace{} & $u$ & $1^{\mathrm{st}}$ two non-0 terms of $\mathrm{arccosh}\, u$ near
$z=x=0$ & Sec.\tabularnewline
\hline
\hline 
\negthinspace{}\negthinspace{}1\negthinspace{}\negthinspace{} & -$2+z^{2}+z^{3}\!$ & $\left(\ln\left(2+\sqrt{3}\right)+\begin{cases}
i\pi & \mathrm{if}\:\Im\left(z^{2}+z^{3}\right)\geq0\\
-i\pi & \mathrm{otherwise}\end{cases}\right)-\dfrac{z^{2}}{\sqrt{3}}+\cdots$ & \ref{sub:AcoshDeg0CfLtMinus1}\tabularnewline
\hline 
\negthinspace{}\negthinspace{}2\negthinspace{}\negthinspace{} & -$2+iz^{1/4}+z$ & $\left(\ln\left(2+\sqrt{3}\right)+i\pi\right)-iz^{1/4}/\sqrt{3}+\cdots$ & \ref{sub:AcoshDeg0CfLtMinus1}\tabularnewline
\hline 
\negthinspace{}\negthinspace{}3\negthinspace{}\negthinspace{} & $\dfrac{1}{2}+z^{2}+z^{3}$ & $\left(\begin{cases}
1 & \mathrm{if}\:\Im\left(z^{2}+z^{3}\right)\geq0\\
-1 & \mathrm{otherwise}\end{cases}\right)\left(\dfrac{i\pi}{3}-\dfrac{2iz^{2}}{\sqrt{3}}\cdots\right)$ & \ref{sub:AcoshDeg0CfInMinus1ZeroOrZero1}\tabularnewline
\hline 
\negthinspace{}\negthinspace{}4\negthinspace{}\negthinspace{} & $1+z^{2}+z^{3}$ & $\left(\begin{cases}
1 & \mathrm{if}\:-\pi<\arg\left(1+z\right)+2\arg z\leq\pi\\
-1 & \mathrm{otherwise}\end{cases}\right)\left(\sqrt{2}z+\dfrac{z^{2}}{\sqrt{2}}+\cdots\right)$ & \ref{sub:AcoshhDeg0Cf1}\tabularnewline
\hline 
\negthinspace{}\negthinspace{}5\negthinspace{}\negthinspace{} & $-1+z^{2}+z^{3}$ & $\begin{array}{c}
\left(\begin{cases}
1 & \Im\left(z^{2}+z^{3}\right)\geq0\\
-1 & \mathrm{otherwise}\end{cases}\right)\left(i\pi+B\cdot\left(-i\sqrt{2}z+\cdots\right)\right)\qquad\\
\qquad\qquad\qquad|\: B=(-1)^{\left\lfloor \dfrac{1}{2}-\dfrac{\arg\left(1+z\right)+2\arg z}{2\pi}\right\rfloor }\end{array}$ & \ref{sub:AcoshDeg0CfMinus1}\tabularnewline
\hline 
\negthinspace{}6\negthinspace{}\negthinspace{} & $z^{-2}+z^{-1}$ & $\left(2i\pi\left\lfloor \dfrac{\pi-\arg\left(1+z\right)+2\arg\left(z\right)}{2\pi}\right\rfloor +\ln2-2\ln z\right)+z+\cdots$ & \ref{sub:AcoshNegDeg}\tabularnewline
\hline 
\negthinspace{}\negthinspace{}7\negthinspace{}\negthinspace{} & $z+z^{2}$ & $\begin{array}{c}
\left((-1)^{D_{1}+D_{2}}\begin{cases}
1 & \mathrm{if}\:\Im\left(z+\cdots\right)\geq0\\
-1 & \mathrm{otherwise}\end{cases}\right)\left(\dfrac{i\pi}{2}-iz+\cdots\right)\\
\mbox{\mbox{}}\\
\mathrm{|}\; D_{1}=\left\lfloor \dfrac{1}{2}-\dfrac{\arg\left(1+z\right)+\arg z}{2\pi}\right\rfloor \qquad\qquad\\
\:\:\wedge\: D_{2}=\left\lfloor \dfrac{1}{2}-\dfrac{\arg\left(-1+z+z^{2}+\cdots\right)-\arg z}{2\pi}\right\rfloor \end{array}$ & \ref{sub:AcoshPositiveDeg}\tabularnewline
\hline 
\negthinspace{}\negthinspace{}8\negthinspace{}\negthinspace{} & $-2-x^{3/4}$ & $\left(\ln\left(2+\sqrt{3}\right)+i\pi\begin{cases}
-1 & \mathrm{if\:}x<0\\
1 & \mathrm{otherwise}\end{cases}\right)+\dfrac{x^{3/4}}{\sqrt{3}}+\cdots$ & \ref{sub:AcoshDeg0CfLtMinus1}\tabularnewline
\hline 
\negthinspace{}\negthinspace{}9\negthinspace{}\negthinspace{} & $-2+x$ & $\left(\ln\left(2+\sqrt{3}\right)+i\pi\right)-ix/\sqrt{3}+\cdots$ & \ref{sub:AcoshDeg0CfLtMinus1}\tabularnewline
\hline 
\negthinspace{}\negthinspace{}10\negthinspace{}\negthinspace{} & $\dfrac{1}{2}-\sqrt{x}$ & $\left(\begin{cases}
1 & \mathrm{if}\: x\geq0\\
-1 & \mathrm{otherwise}\end{cases}\right)\left(\dfrac{\pi i}{3}+\dfrac{2ix^{1/2}}{\sqrt{3}}+\cdots\right)$ & \ref{sub:AcoshDeg0CfInMinus1ZeroOrZero1}\tabularnewline
\hline 
\negthinspace{}\negthinspace{}11\negthinspace{}\negthinspace{} & $\nicefrac{1}{2}+x$ & $\pi i/3-2ix/\sqrt{3}$ & \ref{sub:AcoshDeg0CfInMinus1ZeroOrZero1}\tabularnewline
\hline 
\negthinspace{}\negthinspace{}12\negthinspace{}\negthinspace{} & $1-x^{2}-x^{3}$ & $\left(\begin{cases}
1 & \mathrm{if}\: x\geq0\\
-1 & \mathrm{otherwise}\end{cases}\right)\left(\sqrt{2}ix+\dfrac{i}{\sqrt{2}}x^{2}+\cdots\right)$ & \ref{sub:AcoshhDeg0Cf1}\tabularnewline
\hline 
\negthinspace{}\negthinspace{}13\negthinspace{}\negthinspace{} & $1-x^{2}-x^{3/2}$ & $\sqrt{2}ix+ix^{3/2}/\sqrt{2}+\cdots$ & \ref{sub:AcoshhDeg0Cf1}\tabularnewline
\hline 
\negthinspace{}\negthinspace{}14\negthinspace{}\negthinspace{} & $-1+x^{2}$ & $\pi i+\left(\begin{cases}
1 & \mathrm{if}\: x\geq0\\
-1 & \mathrm{otherwise}\end{cases}\right)\left(-\sqrt{2}ix+\cdots\right)$ & \ref{sub:AcoshDeg0CfMinus1}\tabularnewline
\hline 
\negthinspace{}\negthinspace{}15\negthinspace{}\negthinspace{} & $-x^{-2}-x^{-1}$ & $\left(-2\ln x+\ln2+\begin{cases}
\pi i & \mathrm{if}\: x\geq0\\
3\pi i & \mathrm{otherwise}\end{cases}\right)-\dfrac{3ix}{2}+\cdots$ & \ref{sub:AcoshNegDeg}\tabularnewline
\hline 
\negthinspace{}\negthinspace{}16\negthinspace{}\negthinspace{} & $-x^{-2}+x^{-1/2}$ & $\left(-2\ln x+\ln2+\pi i\right)-x^{3/2}+\cdots$ & \ref{sub:AcoshNegDeg}\tabularnewline
\hline 
\negthinspace{}\negthinspace{}17\negthinspace{}\negthinspace{} & $2x-x^{3/2}$ & $\pi i/2-2ix+\cdots$ & \ref{sub:AcoshPositiveDeg}\tabularnewline
\hline
\end{tabular}
\end{table}

\section*{Summary}

The generalization from Taylor series to generalized Puiseux series
introduces a surprising number of difficulties that haven't been fully
addressed in previous literature and implementations. The most serious
of these is incorrect results for expansion points that are on a branch
cut or a branch point. Formulas are presented here that correct this
for logarithms, fractional powers, inverse trigonometric functions,
and inverse hyperbolic functions. These corrections typically entail
an additive piecewise-constant multiple of $2\pi i$ and/or a unit-magnitude
piecewise-constant multiplicative factor. Some of these corrections
are applicable even to Taylor series. There are many alternative formulas
for these corrections. The alternatives presented here are chosen:\vspace{-0.1in}

\begin{itemize}
\item to reduce catastrophic cancellation when evaluated with approximate
arithmetic,\vspace{-0.1in}

\item to candidly reveal the piecewise constancy,\vspace{-0.1in}

\item to reveal the boundaries of the pieces as explicitly as is practical,\vspace{-0.1in}

\item to simplify significantly where practical, such as for real expansion
variables or numeric coefficients.\vspace{-0.1in}

\item to be as concise as possible subject to the above goals.
\end{itemize}
Tests, formulas and algorithms are given that compute simplified versions
of these correction terms and factors near the expansion point for
both real and complex expansion variables.

\section*{Acknowledgments}

I thank David Diminnie for his assistance.

\end{document}